\documentclass[conference]{IEEEtran}
\IEEEoverridecommandlockouts

\usepackage{tikz} % to be able to draw some self-contained figs
\usepackage{amsmath,amssymb,amsfonts}
\usepackage{graphicx}
\usepackage{textcomp}
\usepackage{bmpsize}
\usepackage{xcolor}
\usepackage{lipsum}
\usepackage[ruled,linesnumbered,vlined]{algorithm2e}
\usepackage{longtable,supertabular,booktabs}
\usepackage{threeparttablex,caption}
\usepackage{makecell}
\usepackage{multirow}
\usepackage{bbding}
\usepackage{wasysym}
\usepackage{iitem}
\usepackage{latexsym}
\usepackage{stmaryrd}
\usepackage{ulem}
\usepackage{array}
\usepackage{booktabs}
\usepackage{pifont}
\usepackage{subcaption}
\usepackage[]{hyperref} 
\usepackage{cleveref}
\usepackage{amsthm}
\usepackage{ulem}
\usepackage{bm}
\usepackage{xspace}
\newcommand{\etal}{et al.\xspace}

\newcommand{\ie}{i.e.,\xspace}
\newcommand{\eg}{e.g.,\xspace}

\let\oldnl\nl% Store \nl in \oldnl
\newcommand{\nonl}{\renewcommand{\nl}{\let\nl\oldnl}}% Remove line number for one line
\newcommand{\sys}{\textsl{\mbox{GTree}}\xspace}
\crefformat{section}{\S#2#1#3} % see manual of cleveref, section 8.2.1
\crefformat{subsection}{\S#2#1#3}
\crefformat{subsubsection}{\S#2#1#3}
\newcommand{\ZZ}{\mathbb{Z}}
\newcommand{\RR}{\mathbb{R}}

\newcommand{\share}[1] {\langle #1 \rangle}

\def\BibTeX{{\rm B\kern-.05em{\sc i\kern-.025em b}\kern-.08em
    T\kern-.1667em\lower.7ex\hbox{E}\kern-.125emX}}
\newtheorem{theorem}{Theorem}

\begin{document}
\title{\sys: GPU-Friendly Privacy-preserving Decision Tree Training and Inference}

\author{
\IEEEauthorblockN{
    Qifan Wang\IEEEauthorrefmark{1}\IEEEauthorrefmark{2},
    Shujie Cui\IEEEauthorrefmark{3},
    Lei Zhou\IEEEauthorrefmark{4},
    Ye Dong\IEEEauthorrefmark{5},
    Jianli Bai\IEEEauthorrefmark{6},
    Yun Sing Koh\IEEEauthorrefmark{1},
    and Giovanni Russello\IEEEauthorrefmark{1}
}
\\
\IEEEauthorblockA{\IEEEauthorrefmark{1}University of Auckland, New Zealand}
\IEEEauthorblockA{\IEEEauthorrefmark{2}University of Birmingham, United Kingdom}
\IEEEauthorblockA{\IEEEauthorrefmark{3}Monash University, Australia} 
\IEEEauthorblockA{\IEEEauthorrefmark{4}National University of Defense Technology, China}
\IEEEauthorblockA{\IEEEauthorrefmark{5}Singapore University of Technology and Design, Singapore}
\IEEEauthorblockA{\IEEEauthorrefmark{6}Singapore Management University, Singapore}
\thanks{Accepted at TrustCom'24, December 17-21, 2024, Sanya, China.}
}

\maketitle

\begin{abstract}
Outsourcing Decision tree (DT) training and inference to cloud platforms raises privacy concerns. 
Recent Secure Multi-Party Computation (MPC)-based methods are hindered by heavy overhead. Few recent studies explored GPUs to improve MPC-protected deep learning, yet integrating GPUs into MPC-protected DT with massive data-dependent operations remains challenging, raising question: \textit{can MPC-protected DT training and inference fully leverage GPUs for optimal performance?}

We present \sys, the first scheme that exploits GPU to accelerate MPC-protected secure DT training and inference. 
\sys is built across 3 parties who jointly perform DT training and inference with GPUs. 
\sys is secure against semi-honest adversaries, ensuring that no sensitive information is disclosed. 
\sys offers enhanced security than prior solutions, which only reveal tree depth and data size while prior solutions also leak tree structure. 
With our \textit{oblivious array access}, access patterns on GPU are also protected. 
To harness the full potential of GPUs, we design a novel tree encoding method and craft our MPC protocols into GPU-friendly versions. 
\sys achieves ${\thicksim}11{\times}$ and ${\thicksim}21{\times}$ improvements in training SPECT and Adult datasets, compared to prior most efficient CPU-based work. 
For inference, \sys outperforms the prior most efficient work by $126\times$ when inferring $10^4$ instances with a 7-level tree. 
\end{abstract}
% \textcolor{orange}{Specifically, they developed a suite of GPU-friendly versions of MPC protocols.}
%However, unlike most deep learning tasks, DT training and inference generally proceed sequentially. 
%We design the array-based procedure to perform training and inference on the GPUs.
%The building blocks such as \textit{equality test} and \textit{oblivious array access} are also amenable to GPU. 
%\notewang{Shall we add the reference and label (Hoogh et al.) to here. I'm worried someone would argue 2014 is too old, although we re-implement it using the latest framework and latest papers demonstrated Hoogh et al. is still the state-of-the-art to process categorical features.}
%\notewang{Shall we add the comparison with TDSC2020 to here?}
% running on CPU. -->ambiguous



% \fi
%%\IEEEpeerreviewmaketitle

\section{Introduction}
\label{sec:intro}
The decision tree (DT) is a powerful and versatile Machine learning (ML) model. Its excellent interpretability makes it a popular choice for various applications, such as medical diagnosis~\cite{azar2013decision} and weather prediction~\cite{wang2018short}. 
To efficiently train a DT and make predictions, a common solution is outsourcing the tasks to cloud platforms. However, this poses privacy risks, as the underlying data may be sensitive.
For privacy-sensitive applications like healthcare, all data samples, trained models, inference results, and any intermediate data generated during the training and inference should be protected from the Cloud Service Provider (CSP).
This need has driven the development of privacy-preserving DT training and inference (PPDT). 

The typical way to develop PPDT is employing cryptographic primitives, such as {Secure Multi-Party Computation (MPC/SMC)}~\cite{abspoel2020secure,adams2022privacy,emekcci2007privacy,hoogh2014practical,lindell2000privacy,samet2008privacy,tueno2019private} and Homomorphic Encryption (HE)~\cite{bost2015machine,wu2016privately,tai2017privacy,akavia2022privacy}, or a combination of both~\cite{liu2020towards}. 
MPC enables multiple CSPs to jointly compute without revealing any party's private inputs, and HE allows the CSP to perform computations over ciphertexts without decryption. 
Nevertheless, most existing approaches are still not secure enough. 
Earlier works~\cite{emekcci2007privacy,lindell2000privacy} focus mainly on the input data privacy while not considering the model. 
Some works~\cite{hoogh2014practical,liu2020towards,samet2008privacy}, leak statistical information and tree structures from which the adversary could infer sensitive information~\cite{adams2022privacy,chatel2021sok}.  
Existing approaches also suffer from heavy computation and communication overheads. 
%Particularly, since HE-based methods are computationally intensive, they may degrade rapidly as the size of the dataset and the complexity of the DT increase~\cite{bost2015machine,wu2016privately,tai2017privacy,akavia2022privacy,liu2020towards}.
For instance, the recent PPDT method proposed in~\cite{liu2020towards} takes over 20 minutes and ${\thicksim}2$ GB communication cost to train a tree of depth 7 with 958 samples and 9 categorical features. 

Some recent works~\cite{ohrimenko2016oblivious,law2020secure,wang2022enclavetree} show that Trusted Execution Environment (TEE) is more efficient for designing PPDT. 
With TEE's protection, the data can be processed in plaintext. 
However, TEEs such as Intel SGX~\cite{costan2016intel} have limited computational power compared with Graphic Processing Units (GPUs). 
Processing all the DT tasks inside the TEE still cannot achieve ideal performance. 
% GPU TEEs such as Graviton \cite{volos2018graviton} require hardware modification which would adversely affect compatibility. 
Recent NVIDIA H100 GPU \cite{h100gpu2022} extending TEE into GPU has a low performance-to-price ratio in most applications due to the high price. 
% Most of them are not easy to use due to their poor applicability and high cost. 

Hardware acceleration is crucial in the evolution of modern ML. % due to the support for highly-parallelizable workloads. 
Recent works~\cite{ng2021gforce,tan2021cryptgpu,watson2022piranha,jawalkar2023orca} have used GPUs to accelerate secure deep learning. 
They can take full advantage of GPU due to massive GPU-friendly operations (\eg convolutions and matrix multiplications). 
This raises the question: \textit{can secure DT training and inference benefit from GPU acceleration?}

\noindent\textbf{Our goals and challenges.}
In this work, we aim to design a GPU-based system to securely train a DT model and make predictions.
Specifically, our approach does not reveal any information other than the input size and tree depth.

GPUs are especially effective when processing grid-like structures such as arrays and matrices~\cite{cuda2022programming}.
To better utilize GPU acceleration, \sys represents the DT and data structures as arrays and performs training and inference on arrays. 
However, processing them securely on GPU is non-trivial because it suffers from access pattern leakage, which enables an adversary to infer tree shape~\cite{ohrimenko2016oblivious,wang2022enclavetree}.
Thus, \textit{the first challenge is to protect access patterns on GPU}.

To protect the data and DT, we employ MPC in training and inference.
However, inappropriate MPC protocols could impede GPU acceleration. 
GPU will be sufficiently used when performing a large number of simple arithmetic computations (\eg addition and multiplication) on massive data in parallel. 
The operations involve a large number of conditional statements, modular reduction, and non-linear functions, e.g., exponentiation and division are less well-suited for GPU. 
For instance, CryptGPU~\cite{tan2021cryptgpu} shows that GPU achieves much less performance gain when evaluating non-linear functions in private neural network training.
Thus, \textit{the second challenge is to design \textit{GPU-friendly} MPC protocols for DT training and inference so as to take full advantage of GPU parallelism}.

\noindent\textbf{Our contributions.}
In this work, we introduce \sys, a GPU-based privacy-preserving DT training and inference addressing the aforementioned challenges. 
Rather than focusing solely on improving MPC protocols, \textit{\sys primarily explores the possibility of combining GPU and MPC-based DT}, as done in recent GPU-based neural network schemes such as CryptGPU~\cite{tan2021cryptgpu}, Piranha~\cite{watson2022piranha}, and Orca~\cite{jawalkar2023orca}. 
To the best of our knowledge, \sys is the first to deploy secure DT training and inference on GPU.
%, which effectively narrows the gap between privacy-preserving DT and practical applications.
Furthermore, \sys achieves better security guarantees than existing solutions. Particularly, \sys conceals all data items,  the tree shape, statistical information, and the access pattern for both DT training and inference.

\sys relies on secret-sharing-based MPC protocols, which are better suited to GPU than garbled circuits (GC)~\cite{yao1986generate}. 
ABY$^3$~\cite{mohassel2018aby3} and Piranha~\cite{watson2022piranha} show that, in the presence of a semi-honest adversary with an honest majority, 3-Party Computation (3PC) with \textit{2-out-of-3 replicated secret sharing (RSS)}~\cite{araki2016high} is more efficient than $2/N$-PC ($N>3$) in runtime and communication cost.
Thus, \sys involves 3 non-colluding CSPs to securely share data and model via \textit{2-out-of-3 RSS}. 
The main contributions are summarized as follows:

\textbf{Access pattern protection.} For both training and inference, access patterns to an array should be protected, even when encrypted. Otherwise, adversaries could deduce sensitive information. For example, data similarities can be inferred from tree access patterns, and if an adversary knows some patterns, they could further uncover more tree information.
We design a GPU-friendly \textit{Oblivious Array Access} protocol which accesses every element of the array and uses a secret-sharing-based select function (\ie $\mathtt{SelectShare}$~\cite{wagh2020falcon}) to ensure that only the desired element is actually read or written.

\textbf{GPU-friendly design.} To make DT training and inference GPU-friendly, our main idea is to parallelize as many operations as possible. 
Firstly, \textcolor{black}{in addition to ensuring that the tree is always complete}, \sys trains layer-wisely based on our novel tree encoding method. 
In doing so, adversaries can only learn the tree depth. 
Such design also allows for extensive parallelization. 
Secondly, all the designed protocols mainly involve simple arithmetic operations and minimized conditional statements so that they are highly parallelable.
%end

During DT training, \sys calculates Gini Index to select the best feature. 
The designed MPC-based protocol uses the scaling function~\cite{watson2022piranha} and secure division~\cite{wagh2020falcon}, potentially leading to accuracy loss.
%However, the comparison of Gini Index values often demands high precision and involves varying precisions, which may not always select the best feature.
To mitigate this, we introduce a TEE-assisted method that enables the oblivious computation of the Gini Index within TEE using oblivious primitives~\cite{ohrimenko2016oblivious}.

\textbf{Extensive experimental evaluations.} We implemented \sys based on GPU-based MPC platform, Piranha~\cite{watson2022piranha}, and evaluated it on a server with 3 NVIDIA Tesla V100 GPUs.
In this work, we focus on processing binary categorical features. 
The results show that \sys takes about 0.31 and 4.32 seconds to train with SPECT and Adult datasets, respectively, which is ${\thicksim}11{\times}$ and ${\thicksim}21{\times}$ faster than the previous work~\cite{hoogh2014practical}
(\sys also provides a stronger security guarantee than it).
We also compare \sys with a TEE-only baseline that processes all tasks within an enclave (baseline uses one of the common TEEs, Intel SGX, as the underlying TEE). 
To train $5{\times}10^4$ data samples with 64 features, \sys outperforms TEE-only by ${\thicksim}27{\times}$. 
For DT inference, \sys excels when the tree has less than 10 levels, surpassing the prior most efficient solution~\cite{kiss2019sok} by ${\thicksim}126\times$ for inferring $10^4$ instances with a depth-7 tree. Compared to TEE-only, \sys is more efficient in inferring a batch of instances when the tree depth is less than 6.

\section{Background}
\label{sec:bg}
% In this section, we present the basic information of DT and the cryptographic primitives used in \sys.
%We finally introduce the oblivious primitives we use in SGX.

\subsection{Decision Tree}
\label{DT_steps}
 A decision tree consists of internal nodes and leaves, where each internal node is associated with a test on a feature, each branch represents the outcome of the test, and each leaf represents a label which is the decision taken after testing all the features on the corresponding path.

%\subsubsection{Training}
%\label{sec:dt training}
\noindent\textbf{Training.}
During training, the DT algorithm learns from a dataset with labeled data samples, where each data contains a set of feature values and a label. 
Basically, building a DT includes 3 steps: (1) \textit{Learning phase}: when a leaf becomes an internal node, partition its data samples into new leaves according to the assigned feature. Then, for each new leaf, count the number of each possible $(value, label)$ pair among its data samples. (2) \textit{Heuristic computation (HC)}: if all data samples on a new leaf have the same label, assign that label to the leaf. Otherwise, calculate heuristic measurements (\sys uses gini index~\cite{quinlan2014c4}) for each feature and choose the best feature to split the leaf. For a feature $s$ and dataset $\bm{D}$, the gini index is defined as follows: $\overline{G}(\bm{D}) = 1 - \sum_{k=1}^{|V_{s_{d-1}}|}(\frac{|\bm{D}_k|}{|\bm{D}|})^2$ and $\overline{G}(\bm{D},s) = {\sum}_{i=0}^{|V_s|-1}\frac{|\bm{D}^{v_{s,i}}|}{|\bm{D}|}\overline{G}(\bm{D}^{v_{s,i}})$, where $|\bm{D}_k|$ is the number of samples containing $k$-th label and $k \in [1,|V_{s_{d-1}}|]$, $|V_s|$ is the number of values of feature $s$ ($s=2$ for binary tree and $s_{d-1}$ is the label) and $\bm{D}^{v_{s,i}}$ represents a subset of $\bm{D}$ containing feature $s$ with value of $v_{s,i}$.
(3) \textit{Node split (NS)}: convert the leaf into an internal node using the best feature.

\noindent\textbf{Inference}
To infer an unlabelled data sample (or instance), from the root node, we check whether the corresponding feature value of the instance is $v_0$ or $v_1$, and continue the test with its left or right child respectively, until a leaf node. 
The label value of the accessed leaf is the inference result.

\subsection{Secret Sharing}

In \sys, an $l$-bit value $x{\in}\ZZ_{2^l}$ is secretly shared among three parties (say $P_1, P_2$, $P_3$) with 2-out-of-3 replicated secret sharing (RSS)~\cite{araki2016high}.
\sys leverages \textit{Arithmetic/Boolean sharing}, denoted as $\share{x}^A$ and $\share{x}^B$, respectively. 
Without special declaration, we compute in $\ZZ_{2^l}$ and omit (mod $2^l$) for brevity.

\noindent\textbf{Arithmetic Sharing.} 
We define the following operations: 
\begin{itemize}
    \item $\mathtt{Share}^A(x)\rightarrow (x_1, x_2, x_3)$: Given input $x \in \ZZ_{2^l}$, it samples $x_1,x_2,x_3{\in}\ZZ_{2^l}$ such that $x_1+x_2+x_3=x$. $P_i$ holds $(x_i,x_{i+1})$.  
    We denote $\share{x}^A=(x_1, x_2, x_3)$.

    \item $\mathtt{Rec}^A(\share{x}^A)\rightarrow x$: Given input $\share{x}^A=(x_1, x_2, x_3)$, it outputs $x= x_1 + x_2 + x_3$. 
    $P_{i}$ receives $x_{i-1}$ from $P_{i-1}$ and reconstructs $x$ locally.

    \item \textit{Linear operation}: 
    Given public constants ${\alpha},{\beta}~{\in}~\ZZ_{2^l}$, $P_i$ computes their respective shares of $\share{{\alpha}x+{\beta}}^A = ({\alpha}x_1+{\beta}, {\alpha}x_2, {\alpha}x_3)$ locally. 

    \item \textit{Multiplication}: $\share{z}^A = \share{x}^A \cdot \share{y}^A$: $z_i = x_iy_i + x_{i+1}y_i + x_iy_{i+1}$. With $(x_i,x_{i+1})$ and $(y_i,y_{i+1})$, $P_i$ compute $z_i$ locally and get the additive shares of $z$. 
    To ensure parties hold replicated shares of $z$, $P_i$ sends $P_{i+1}$ a blinded share $z_i+{\alpha}_i$, where ${\alpha}_1,{\alpha}_2,{\alpha}_3{\in}\ZZ_{2^l}$ and ${\alpha}_1 + {\alpha}_2 + {\alpha}_3 = 0$~\cite{mohassel2018aby3}. 
\end{itemize}

\noindent\textbf{Boolean Sharing.} 
Boolean sharing uses an XOR-based secret sharing. 
For simplicity, given $l$-bit values, we assume each operation is performed $l$ times in parallel.
The semantics and operations are the same as the arithmetic shares except that $+$ and $\cdot$ are respectively replaced by bit-wise $\oplus$ and $\wedge$.

\section{Threat Model of \sys}
\label{sec:overview}

\sys consists of 3 types of entities: Data Owner (DO), Query User (QU), and Cloud Service Provider (CSP). 
\sys employs 3 independent CSPs, such as Facebook, Google, and Microsoft, and they should be equipped with GPUs. 
The DO outsources data samples to the three CSPs securely with RSS
% the 2-out-of-3 secret sharing technique 
for training the DT, and the QU queries after the training for data classification or prediction.

% \noindent\textbf{Threat Model.}
%\subsection{Threat Model}
Following with other three-party-based systems~\cite{mohassel2018aby3,tan2021cryptgpu,wagh2020falcon}, \sys resists a semi-honest adversary with an honest majority among three CSPs.
Specifically, two of them are fully honest, and the third one is semi-honest, \ie follows the protocol but may try to learn sensitive information, \eg tree structure, by analyzing the access patterns and other leakages. The DO and QU are fully trusted. 
We assume three CSPs possess shared point-to-point communication channels and pairwise shared seeds, which are used by AES as a PRNG to generate secure randomness~\cite{wagh2020falcon}.
Attacks~\cite{naghibijouybari2018rendered,jiang2019memory} show that GPUs can be shared between multiple applications, enabling spy applications to monitor and infer the behaviour of victims. 
For instance, recent work~\cite{khairy2020accel} introduced a trace tool to observe the memory traces of the application running on the GPU. 
In this work, we assume adversaries possess this ability to mount similar attacks on GPU and observe the victim's access pattern over the GPU's on-chip memory.
% As the QU gets the predictions in the clear, \sys does not protect the privacy of training data from attacks such as membership inference~\cite{shokri2017membership}.
% Defending against these attacks is an orthogonal problem and is out of the scope of this work.
Additionally, \sys does not protect against denial-of-service attacks and other attacks based on physical channels.

In the sub-protocol using TEE, the TEE enclave can reside on any of the three CSPs. 
The TEE enclave is fully trusted. Components beyond the enclave are considered untrusted, as they may infer secrets by observing the memory access patterns. 
In setup, secure channels are established between three CSPs and the enclave using secret keys. 
All shares transmitted between the CSP and the enclave are encrypted with their respective keys, employing AES-GCM.

The security proof can be found in Section~\ref{sec:security}. 
\section{Data and Model Representation}
\label{sec:data_model}

\begin{table}
	\centering
	\footnotesize
	\caption{Notations}
	\label{tab:notation}
	\begin{tabular}{c l}
		\toprule
		\textbf{Notation} & \textbf{Description} \\ \hline
		$\bm{D}$ & Data samples array where $|\bm{D}|=N_D$ \\	\hline
		$d$ & Number of features \\	\hline
		$\bm{S}$ & A sequence of $d$ features, $\bm{S}=(s_0, s_1, {\cdots}, s_{d-1})$ \\	\hline
		$\bm{V}_{s_i}$ & {\makecell[l]{Values of feature $s_i$, $\bm{V}_{s_i}=(v_{i,0}, v_{i,1})$ (binary)}} \\	\hline
            $n_h$ & Number of nodes at level $h$, $n_h=2^h$ \\ \hline
            ${\Psi}_{r_{i}}$ & \makecell[l]{Number of data samples at a leaf containing the $i$-th \\label, where $i{\in}[0,1]$} \\ \hline
		$\overline{G}(\cdot)$ & {\makecell[l]{Heuristic measurement, \eg Gini Index}} \\  \hline
            $\bm{T}$ & {\makecell[l]{Tree array stores tree nodes (\ie internal, leaf and \\ dummy nodes). $\bm{T}[2i+1]$ and $\bm{T}[2i+2]$ are the left \\ and right children of $\bm{T}[i]$, respectively}} \\ \hline
            $\bm{F}$ & {\makecell[l]{Node type array indicates the node type of $\bm{T}[i]$ where \\ $|\bm{F}|= |\bm{T}|$ and $\bm{F}[i]=0,1,2$ corresponds to internal, \\ leaf and dummy node}} \\ \hline
            $\bm{M}$ & {\makecell[l]{An array indicating the leaf each data sample $\bm{D}[i]$ \\ currently belongs to, where $|\bm{M}|=|\bm{D}|$}} \\ \hline
            $c_{(value, label)}$ & Frequency of each $(value, label)$ \\	\hline
            $\bm{C}$ & {\makecell[l]{2D-counter array has 3 rows and $2(d-1)$ columns, \\ storing the frequency of each feature}} \\ \hline
            $\bm{\gamma}$ & {\makecell[l]{An array indicating if each feature has been assigned \\ to a leaf node or not, where $|\bm{\gamma}|=d-1$}} \\
		\bottomrule
	\end{tabular}
\end{table}

This section presents data and DT representation in \sys. 
In the rest of this paper, we will use the notation in Table~\ref{tab:notation}.

% \noindent\textbf{Data Representation.}
\subsection{Data Representation}
We stress that \sys focuses on training categorical features. 
We assume that $\bm{S}=(s_0,s_1,{\cdots},s_{d-1})$ represents a sequence of $d$ features, where each feature $s_i$ having two possible values (\ie binary features): $\bm{V}_{s_i}=(v_{i,0},v_{i,1})$. 
In \sys, we convert $v_{i, 0}$ and $v_{i, 1}$ to 0 and 1, respectively, for all $i \in [0, d-1]$. %. 
%$V_{s_i}=(v_{i,1},v_{i,2})$
A labeled data sample is represented with $d$ feature values, \eg $\{1, 0, ..., 1\}$, and the last feature value is the label, and unlabelled data samples only have $d-1$ feature values. 
For clarity, we use $r_0$ and $r_1$ to represent the two label values. 
We denote an array of $N_D$ data samples as $\bm{D}$, and each data sample $\bm{D}[i]$ is an array containing $d$ or $d-1$ elements, where  $0{\leq}i{\leq}N_D-1$. 
%, where $D[i][j]$ is its $j$-th element, and $1{\leq}j{\leq}d$.
The DO computes $\mathtt{Share}^A(\bm{D})$ and distributes replicated shares to the 3 CSPs. 

\begin{figure}
\centering
\includegraphics[width=0.95\linewidth]{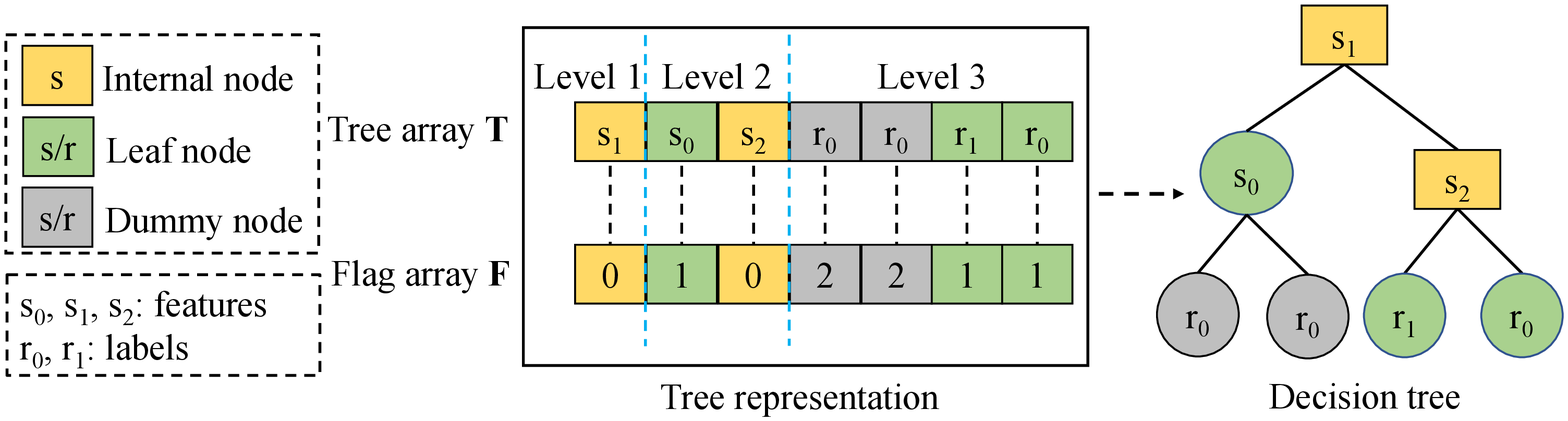}
\caption{Representation of an example tree in \sys. \textnormal{The tree is represented with two arrays $\bm{T}$ and $\bm{F}$, where $|\bm{T}|=|\bm{F}|$. $\bm{T}[i]$ is an internal, leaf, or dummy node, corresponding to $\bm{F}[i]=0,1$, or $2$, respectively. $\bm{T}[i]$ is the assigned feature if it is an internal node, \eg $\bm{T}[0]=s_1$ and $\bm{T}[2]=s_2$. For the dummy and leaf nodes in the last level, $\bm{T}[i]$ is the label of the path, \eg $\bm{T}[3]=r_0$ and $\bm{T}[5]=r_1$; otherwise, $\bm{T}[i]$ is a random feature, \eg $\bm{T}[1]=s_0$ (it is indeed a leaf). Although $\bm{T}[3]$ and $\bm{T}[4]$ are dummy nodes, they play the role of a leaf and store the label of their paths since they are in the last level. Such design not only protects tree shape but also allows for training and inference in a highly parallelized manner.}}
\label{fig:tree encoding}
\end{figure}

\subsection{Model Representation}

Our tree encoding method is presented in Fig.~\ref{fig:tree encoding}. Such design enables \sys to train and infer data in a highly parallelised manner. 
By padding the model into a complete binary tree with dummy nodes, all the paths have the same length. As a result, no matter which path is taken during the training and inference, the processing time will always be the same. 
Since the model is a complete binary tree, the size of $T$ is determined by the depth. The depth of the tree is normally unknown before the training, yet $|\bm{T}|$ should be determined. To address this issue, we initialize $|\bm{T}|$ with the maximum value, \ie $|\bm{T}|=2^d-1$. In case the tree depth $h$ is predefined, $|\bm{T}|=2^h -1$. 
$T$ is initialized with $|\bm{T}|$ leaves. 
Both $T$ and $F$ are shared among the three CSPs with RSS.

\section{Construction Details of \sys}
\label{sec:all protocols}
In this section, we first present one basic block about how CSPs access arrays without learning access patterns.
We then describe how \sys obliviously performs 3 steps of DT training: \textit{learning phase, heuristic computation, and node split}. 
Finally, we present the details of secure DT inference. 
% In the rest of this paper, we will use the symbols shown in Table~\ref{tab:acronyms_algo} to represent the protocols designed for \sys.

\subsection{Oblivious Array Access}
\label{sec:protocol:oaa}
To take advantage of GPU's parallelism, \sys represents both the tree and data samples with arrays. 
Although they are all protected under RSS, the CSPs can observe the access patterns.
There are two common techniques to conceal access patterns: linear scan and ORAM.
Considering the basic construction of linear scan involves a majority of highly parallelizable operations like vectorized arithmetic operations, we construct our Oblivious Array Access protocol ${\prod}_{\rm OAA}$ based on linear scan, where three CSPs jointly scan the whole array and obliviously retrieve the target element.

\noindent{\textbf{Equality test:}} ${\prod}_{\rm OAA}$ needs to obliviously check the equality between two operands or arrays with their shares.  
We modified the \textit{Less Than} protocol in Rabbit~\cite{makri2021mathsf} to achieve that, denoted as $\share{e}^B \leftarrow {\prod}_{\rm EQ}(\share{x}^A, \share{y}^A)$. Specifically, $e=1$ if $x=y$; otherwise $e=0$.
It also applies to arrays, \ie $\share{\bm{E}}^B \leftarrow {\prod}_{\rm EQ}(\share{\bm{X}}^A, \share{\bm{Y}}^A)$. $\bm{E}[i]=1$ if $\bm{X}[i]=\bm{Y}[i]$, otherwise $\bm{E}[i]=0$, where $i \in [0, |\bm{X}|-1]$. 
Notably, one of the inputs can be public (\eg $y$ or $\bm{Y}$).

\normalem
\begin{algorithm}[htbp]
\footnotesize
\DontPrintSemicolon 
\caption{Oblivious Array Access, ${\prod}_{\rm OAA}$}
\label{alg2}
\SetKwInOut{Input}{Input}\SetKwInOut{Output}{Output}

\Input{$\share{\bm{W}}^A$ and $\share{\bm{U}}^A$ (stores the target indices)}
% of length $N$ of length $W$
\Output{All parties learn the shares of target elements $\share{\bm{Z}}^A=\{\share{\bm{W}[u]}^A\}_{\forall u \in \bm{U}}$}

Parties initialize an array $\share{\bm{O}}^A$ with $|\bm{W}|$ zero values\; \label{alg2:line3}
\For{each $\share{u}^A \in \share{\bm{U}}^A$}{\label{alg2:line1}

    % \For{$k=\{0,\cdots,|\bm{W}|-1\}$}{ \label{alg2:line2}
    %     $\share{c_k}^B={\prod}_{\rm EQ}(\share{u}^A, k)$\; \label{alg2:line3}
    % }
    $\share{\bm{E}}^B \leftarrow {\prod}_{\rm EQ}(\{\share{u}^A\}^{|\bm{W}|}, \{0,{\cdots},|\bm{W}|-1\})$\; \label{alg2:line2}
    $\share{\bm{O}}^A \leftarrow \mathtt{SelectShare}(\share{\bm{O}}^A,\share{\bm{W}}^A,\share{\bm{E}}^B)$\; \label{alg2:line4}
    Compute $\share{z_u}^A = {\sum}_{k=0}^{|\bm{W}|-1} \share{\bm{O}[k]}^A$\; \label{alg2:line5}
}
Output $\share{\bm{Z}}^A = \{\share{z_u}^A\}_{\forall u{\in}\bm{U}}$\; \label{alg2:line6}
\end{algorithm}

\noindent{\textbf{Oblivious Array Access:}} 
Algorithm~\ref{alg2} illustrates our ${\prod}_{\rm OAA}$. 
Given the shares of the array $\bm{W}$ and the shares of a set of target indices $\bm{U}$, it outputs the shares of all the target elements. 
When $|\bm{U}| >1$, the target elements are accessed independently. 
Here we introduce how the three parties access one target element obliviously. 
Assume the target index is $\share{u}^A$, parties first compare it with each index of $\bm{W}$ using ${\prod}_{\rm EQ}$ and get Boolean shares of the $|\bm{W}|$ comparison results: $\share{\bm{E}}^B$ (Line~\ref{alg2:line2}).
%$\{\share{c_0}^B,{\cdots},\share{c_{|\bm{W}|-1}}^B\}$ (Line~\ref{alg2:line3}). 
According to ${\prod}_{\rm EQ}$, only $\bm{E}[u]=1$, whereas that is hidden from the three parties as they only have the shares of $\bm{E}$. 
% For $\mathtt{SelectShare}$, here we use the version in Piranha~\cite{watson2022piranha}. In details, 
Second, the three parties use $\mathtt{SelectShare}$~\footnote{$\share{\bm{W_3}}^A\leftarrow \mathtt{SelectShare}(\share{\bm{W_1}}^A, \share{\bm{W_2}}^A, \share{\bm{I}}^B)$ is a function which select shares from either array $\share{\bm{W_1}}^A$ or $\share{\bm{W_2}}^A$ based on $\share{\bm{I}}^B$. Specifically, $\bm{W_3}[i]=\bm{W_1}[i]$ when $\bm{I}[i]=0$, and $\bm{W_3}[i]=\bm{W_2}[i]$ when $\bm{I}[i]=1$ for all $i \in [0, |\bm{I}|-1]$.} to obliviously select elements from two arrays $\bm{W}$ and $\bm{O}$ based on $\share{\bm{E}}^B$(Line~\ref{alg2:line4}), where $\bm{O}$ is an assistant array and contains $|\bm{W}|$ zeros. 
Since only $\bm{E}[u]=1$, after executing $\mathtt{SelectShare}$, we have $\bm{O}[u]=\bm{W}[u]$ and all the other elements of ${\bm{O}}$ are still 0. 
Finally, each party locally adds all the elements in $\share{\bm{O}}^A$ and gets  $\share{z_u}^A$, where $z_u=\bm{W}[u]$ (Line~\ref{alg2:line5}).   
% This algorithm can run in parallel with $|\bm{U}|*|\bm{W}|$ threads. Similarly, 
All the for-loops (and in the following protocols) can be parallelized since array values are handled independently. 

%$N{\cdot}\frac{l^2+10l}{8}$ 
% + 2l
To fully leverage GPUs, \sys mainly exploits the parallelism of protocols. The communication overhead is the performance bottleneck of \sys. 
Precisely, given $N$ inputs, ${\prod}_{\rm EQ}$ totally incurs $N{\cdot}(2l-1)$ bits communication overhead in ${\log}l+1$ rounds. 
The communication overhead of ${\prod}_{\rm OAA}$ is $N{\cdot}(4l-1)$ bits in ${\log}l+3$ rounds. 
We can further improve this by using more communication-efficient primitives, \eg ORAM~\cite{doerner2017scaling}. We will explore that in future work.

\subsection{Oblivious DT Training}
We summarise 3 steps for DT training in Section \ref{DT_steps}. Here we present how \sys constructs oblivious protocols for them.
\normalem
\begin{algorithm}[!ht]
 \footnotesize
\DontPrintSemicolon 
\caption{Oblivious Learning, ${\prod}_{\rm OL}$}
\label{alg3}
\SetKwInOut{Input}{Input}\SetKwInOut{Output}{Output}

\Input{$\share{\bm{D}}^A$, $\share{\bm{T}}^A$, $\share{\bm{F}}^A$, $\share{\bm{M}}^A$, $\{\share{\bm{C_n}}^A\}_{n{\in}[0,n_h-1]}$ and current tree level $h$, where $n_h=2^h$ is the number of leaves at level $h$}
\Output{Updated $\{\share{\bm{C_n}}^A\}_{n{\in}[0,n_h-1]}$ }

% \nonl \emph{\% \notewang{here, we need to check if only one root node. if yes, it means there is no node split before data partition, all M[] are 0 (initial)}.}\; 
\If{$h \neq 0$}{ \label{alg3:line1}
    $\share{\bm{Tval}}^A \leftarrow {\prod}_{\rm OAA}(\share{\bm{T}}^A, \share{\bm{M}}^A)$\; \label{alg3:line2}
    \For{$i=\{0,\cdots,N_D-1\}$} { \label{alg3:line3}
        $\share{\bm{Dval}[i]}^A \leftarrow  {\prod}_{\rm OAA}(\share{\bm{D}[i]}^A, \share{\bm{Tval}}^A)$\; \label{alg3:line4}
    }
   Compute $\share{\bm{M}}^A = 2 \share{\bm{M}}^A + \share{\bm{Dval}}^A + 1$\; \label{alg3:line5}
}

Initialize $\bm{LCidx} \leftarrow \{0,{\cdots},n_h-1\}$\; \label{alg3:line6}
% \For{$i=\{0,\cdots,n_h-1\}$} { \label{alg3:line7}
%     $\share{\bm{isLeaf[i]}}^B={\prod}_{\rm EQ}(\share{\bm{F}[i+2^h-1]}^A,1)$\; \label{alg3:line8} 
% }
$\share{\bm{isLeaf}}^B \leftarrow {\prod}_{\rm EQ}(\{\share{\bm{F}[2^h-1]}^A,{\cdots},\share{\bm{F}[n_h+2^h-1]}^A\},\{1\}^{n_h})$\; \label{alg3:line7} 

\For{$i=\{0,\cdots,N_D-1\}$} { \label{alg3:line8}
    % \For{$j=\{0,\cdots,n_h-1\}$} { \label{alg3:line10}
    %     $\share{\bm{LCF}[j]}^B={\prod}_{\rm EQ}(\share{\bm{M}[i]}^A-(2^h-1)), \bm{LCidx_{plain}}[j])$ \; \label{alg3:line11}
    % }
    $\share{\bm{LCF}}^B \leftarrow {\prod}_{\rm EQ}(\{\share{\bm{M}[i]}^A-(2^h-1)\}^{n_h}, \bm{LCidx}$)\; \label{alg3:line9}
    $\share{\bm{LCF}}^B~{\wedge}=~\share{\bm{isLeaf}}^B$\;\label{alg3:line10}
    % Parties initialize \sout{a} \textcolor{orange}{$n_h$ arrays of} $\bm{C}$ for each leaf: $\share{\bm{C}_{D_i}}^A=\{\share{\bm{C_j}}^A\}_{j{\in}[0,n_h-1]}$\;\label{alg3:line13}
    Initialize $n_h$ arrays: $\share{\bm{C'}_{D_i}}^A \leftarrow \{\share{\bm{C'_{\rho}}}^A\}_{{\rho}{\in}[0,n_h-1]}$\;\label{alg3:line11}
    %, where $\bm{C_j}$ is same as $\bm{C}$ but with all zeros.
    \For{$\rho=\{0,\cdots,n_h-1\}$} { \label{alg3:line12}
        \For{$k=\{0,\cdots,d-2\}$}{ \label{alg3:line13}
            $\share{\bm{C'_{\rho}}[0][2k]}^A = (1-\share{\bm{D}[i][k]}^A)$\; \label{alg3:line14}
            $\share{\bm{C'_{\rho}}[0][2k + 1]}^A =  \share{\bm{D}[i][k]}^A$\; \label{alg3:line15}
            $\share{\bm{C'_{\rho}}[1][2k]}^A = (1-\share{\bm{D}[i][k]}^A) {\cdot} (1-\share{\bm{D}[i][d-1]}^A)$\;\label{alg3:line16}
            $\share{\bm{C'_{\rho}}[1][2k+1]}^A = \share{\bm{D}[i][k]}^A {\cdot} (1-\share{\bm{D}[i][d-1]}^A)$\;\label{alg3:line17}
            % \nonl \emph{\% Update 3rd row.}\;
            $\share{\bm{C'_{\rho}}[2][2k]}^A =  (1-\share{\bm{D}[i][k]}^A) {\cdot} \share{\bm{D}[i][d-1]}^A$\;\label{alg3:line18}
            $\share{\bm{C'_{\rho}}[2][2k+1]}^A = \share{\bm{D}[i][k]}^A {\cdot} \share{\bm{D}[i][d-1]}^A$\; \label{alg3:line19}
        }
    }
    $\share{\bm{C'}_{D_i}}^A {\leftarrow} \mathtt{SelectShare}(\{\share{0}^A\}^{n_h},\share{\bm{C'}_{D_i}}^A,\share{\bm{LCF}}^B)$\; \label{alg3:line20}
    %\textcolor{orange}{Here I use $C$ as a unit, so the size is $n_h$.}
}
Accumulate all $\{\share{\bm{C'}_{D_i}}^A\}_{i \in [0,N_D-1]}$ to $\{\share{\bm{C_n}}^A\}_{n{\in}[0,n_h-1]}$.\; \label{alg3:line21}
\end{algorithm}

% \noindent\textbf{Oblivious Learning: Data Partition.}
\subsubsection{\textbf{Oblivious Learning: Data Partition}}
\label{sec:protocol:ol_dp}
Algorithm~\ref{alg3} illustrates details of the learning phase. 
The first step is to partition the data samples into new leaves (Line~\ref{alg3:line1}-\ref{alg3:line5}). 
We need to conceal data distribution, ensuring that the execution flow remains independent of the input data. 
Specifically, \sys uses an array $\bm{M}$ to indicate the leaf each data sample currently belongs to, where $\bm{M}[i]$ is the leaf identifier of $\bm{D}[i]$ and $|\bm{M}|=|\bm{D}|$, \eg  $\bm{M}[i]=n$ if $\bm{D}[i]$ is at $\bm{T}[n]$. 
Assume the new internal node is $\bm{T}[n]$ (which was a leaf with previously partitioned data samples) and the feature assigned to it is $s_j$. 
The main operation of data partition is to check if $\bm{D}[i][j]=0$ for each data sample at $\bm{T}[n]$. If yes, $\bm{D}[i]$ should be partitioned to the left child and $\bm{M}[i] \leftarrow 2\bm{M}[i]+1$; otherwise $\bm{M}[i] \leftarrow 2\bm{M}[i]+2$. That is, $\bm{M}[i] \leftarrow 2\bm{M}[i]+ \bm{D}[i][j] + 1$ (Line~\ref{alg3:line5}). 
To avoid leaking which leaf each data sample belongs to, \sys performs it obliviously. 

\sys processes the leaves at the same level in parallel.  
So for the data partition, we first get the features assigned to new internal nodes at level $h-1$ from $\bm{T}$, the identifiers of which are stored in $\bm{M}$ as they previously were leaves. Parties thus run ${\prod}_{\rm OAA}(\share{\bm{T}}^A, \share{\bm{M}}^A)$ to get all such features (Line~\ref{alg3:line2}). 
The second step is to get the corresponding feature values from each data sample (Line~\ref{alg3:line4}), and then assign $\bm{M}$ with the new leaf identifiers based on the fetched feature values (Line~\ref{alg3:line5}). 
Notably, the values of $\bm{M}$ should be 0 if $h=0$.

\begin{figure}[!ht]
\centering
\includegraphics[width=0.98\linewidth]{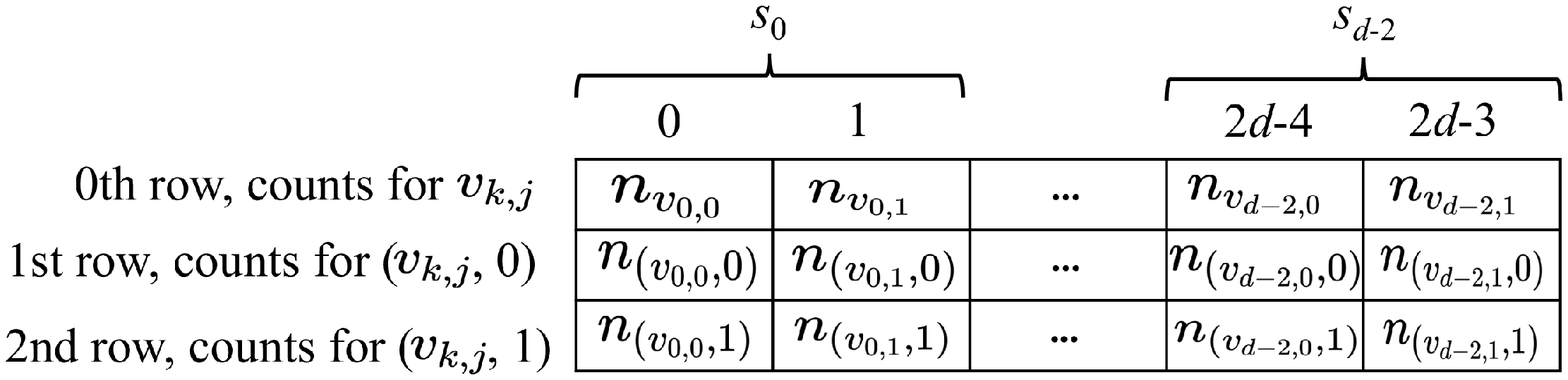}
\caption{Counter array in \sys, denoted as $\bm{C}$. \textnormal{The array contains 3 rows and $2(d-1)$ columns. $C[0][2k+j]$ in the first row stores the number of data samples containing $v_{k,j}$, where $k\in[0, d-2]$ and $j\in\{0, 1\}$. In the last two rows, $\bm{C}[1][2k+j]$ and $\bm{C}[2][2k+j]$ stores counts of $(v_{k,j}, 0)$ and $(v_{k,j}, 1)$ pairs, respectively. Such design also helps to take better advantage of GPU's parallelism.}}
\label{fig:other array}
%\vspace*{-.4cm}
\end{figure}

% \noindent\textbf{Oblivious Learning: Statistics Counting.} 
\subsubsection{\textbf{Oblivious Learning: Statistics Counting}}
\label{sec:protocol:ol_sc}
Statistics information of the data samples at each leaf is required for heuristic measurement. 
Specifically, for each real leaf, it requires the number of data samples containing each feature value and the number of data samples containing each (value, label) pair. 
\sys uses a 2D-array $\bm{C}$ to store such statistics for each leaf. 
As shown in Fig. \ref{fig:other array}, 
$\bm{C}$ has 3 rows and $2(d-1)$ columns (\ie the total number of feature values). 
Particularly, the $2k$-th and $(2k+1)$-th elements of each row are the statistics of feature $s_k$, where $0{\leq}k{\leq}(d-2)$. 
Each element in the first row $\bm{C}[0][2k+j]$ stores the number of data samples containing the corresponding feature value $v_{k, j}$, where $j{\in}\{0,1\}$. 
In the last two rows, $\bm{C}[1][2k+j]$ and $\bm{C}[2][2k+j]$ stores counts of $(v_{k,j}, 0)$ and $(v_{k,j}, 1)$ pairs, respectively. 
%, where $v_d{\in}\{0,1\}$ represents the label.

After assigning data samples to new leaves, \sys updates $\bm{C}$ of all leaves by scanning the partitioned data samples (Line~\ref{alg3:line6}-\ref{alg3:line21}). 
To hide the data samples' distribution over leaves (\ie the statistical information), \sys obliviously updates $\bm{C}$ of all leaves when processing each $\bm{D}[i]$. 
Indeed, only the real leaf that really contains $\bm{D}[i]$ should update its $\bm{C}$. 
To ensure the correctness of the final $\bm{C}$, a flag array $\bm{LCF}$ is used to indicate which $\bm{C}$ should be updated, and $\bm{LCF}$ contains two points: which leaf $\bm{D}[i]$ belongs to (Line~\ref{alg3:line9}), and if this leaf is real (Line~\ref{alg3:line7}, \ref{alg3:line10}). 
Whether each element of $\bm{C}$ increases 1 or not depends on the feature values of $\bm{D}[i]$ (see Line~\ref{alg3:line14}-\ref{alg3:line19}), and \sys stores such information in a temporary counter array $\bm{C'}_{\rho}$ and finally accumulates them in batch (Line~\ref{alg3:line21}). 
Note that there is no need to recompute the multiplications in Line~\ref{alg3:line17} and~\ref{alg3:line19} as they have been previously calculated. 
% when all the data samples are processed

\subsubsection{\textbf{Oblivious Heuristic Computation}}
\label{sec:protocol:ohc}

We next need to convert each leaf into an internal node using the best feature. 
%In DT training, a feature could be assigned to multiple internal nodes on different paths. 
For each path, a feature can only be assigned once, and \sys uses an array $\bm{\gamma}$ of size $d-1$ to indicate if each feature has been assigned or not. Specifically, $\bm{\gamma}[i]=0$ if $s_i$ has been assigned; otherwise $\bm{\gamma}[i]=1$.

We adopt gini index due to its integer-arithmetic-friendly operations, \eg addition and multiplication. 
The Equation in Section~\ref{DT_steps} can be distilled into the following equation based on $\bm{C}$, where $m_k=\bm{C}[k+1][2i+v_{i,0}]+\bm{C}[k+1][2i+v_{i,1}]$:
\begin{equation}
% \footnotesize
\label{eq_gini_reduction}
    \begin{split}
         &\overline{G}(\bm{C},s_i) = {\sum}_{j=0}^{|\bm{V_{s_i}}|-1}\frac{P}{Q} \\
         &P = (\bm{C}[0][2i+v_{i,j}])^2-{\sum}_{k=0}^{|\bm{V}_{s_{d-1}}|-1}{m_k}^2 \\
         &Q = \bm{C}[0][2i+v_{i,j}]{\cdot}(\bm{C}[0][2i+v_{i,0}]+\bm{C}[0][2i+v_{i,1}])
    \end{split}
\end{equation}

We represent gini index as a rational number, where $P$ and $Q$ are the numerator and denominator, respectively.
We observe that minimizing the gini index is equivalent to minimizing $\frac{P}{Q}$.
To find the best feature, we need to select the feature with the smallest gini index value by comparing $\frac{P}{Q}$ of different features.
Na\"ively, we can avoid expensive division by comparing the fractions of two features directly~\cite{hoogh2014practical}: given non-negative $a,b,c,d$, we have that $\frac{a}{b}<\frac{c}{d}$ iff $a{\cdot}d<b{\cdot}c$.
Nevertheless, such a solution imposes the restrictions on the modulus $2^l$ and the dataset size. 
Specifically, given all features are binary, the numerator and denominator are upper bounded by $N_D^3/8$ and $N_D^2/4$, respectively.
The modulus must be at least $5{\cdot}(\log N_D - 1)$ bits long.
Therefore, $N_D$ can be at most $2^{13}=8,192$ for the modulus $2^{64}$
% a $64$-bit integer (\ie the modulus is $2^{64}$)
, which means we can process at most $8,192$ samples.
% Moreover, modulus $2^{32}$ is enough for most cases in \sys. 
Moreover, a larger modulus results in performance degradation.

To our knowledge, this restriction is still an open obstacle to secret-sharing-based secure DT training.
Recently, Abspoel \etal~\cite{abspoel2020secure} use the random forest instead of DT to bypass this limitation.
To avoid the limitation of dataset size, we provide an MPC-only solution, namely \textit{MPC-based HC}: ${\prod}_{\rm OHC}^{mpc}$.

\normalem
\begin{algorithm}[!ht]
\footnotesize
\DontPrintSemicolon 
\caption{Oblivious Heuristic Computation via MPC, ${\prod}_{\rm OHC}^{mpc}$ }
\label{alg_mpc}
\SetKwInOut{Input}{Input}\SetKwInOut{Output}{Output}

\Input{$\{\share{\bm{C_n}}^A\}_{n{\in}[0,n_h-1]}$, $\{\share{\bm{\gamma_n}}^A\}_{n{\in}[0,n_h-1]}$, $\share{\bm{F}}^A$ and current tree level $h$}
\Output{Split decisions array $\share{\bm{SD}}^A$, updated $\{\share{\bm{\gamma_n}}^A\}_{n{\in}[0,n_h-1]}$, $\share{\bm{F}}^A$}

\For{$\rho=\{0,\cdots,n_h-1\}$}{ \label{alg_mpc:line1}
    \For{$i=\{0,\cdots,d-2\}$}{ \label{alg_mpc:line2}
        Compute $\share{\bm{P}[i]}^A$ and $\share{\bm{Q}[i]}^A$ using $\share{\bm{C_t}}^A$ according to Equation~\ref{eq_gini_reduction} \; \label{alg_mpc:line3}
        $\share{\bm{gini}[i]}^A \leftarrow \mathtt{Division}(\share{\bm{P}[i]}^A,\share{\bm{Q}[i]}^A)$ \; \label{alg_mpc:line4}
    }
    $\share{\bm{gini}}^A \leftarrow \mathtt{SelectShare}(\share{\bm{O}}^A,\share{\bm{gini}}^A,{\prod}_{\rm EQ}(\share{\bm{\gamma_{\rho}}}^A, \{1\}^{d-1}))$ \; \label{alg_mpc:line5}
    \nonl \textcolor{gray}{// return the indices of the best features}\;
    $\share{\bm{SD}[\rho]}^A \leftarrow \mathtt{Maxpool}(\share{\bm{gini}}^A)$ \; \label{alg_mpc:line6}
    \nonl \textcolor{gray}{// selectively update $\bm{\gamma_{\rho}}$, $\bm{F}$. $\bm{O}$ is initialized with zeros}\;
    $\share{\bm{\gamma_{\rho}}}^A \leftarrow \mathtt{SelectShare}(\share{\bm{\gamma_{\rho}}}^A,\share{\bm{O}}^A,{\prod}_{\rm EQ}(\share{\bm{\gamma_{\rho}}}^A, \{\share{\bm{SD}[\rho]}^A\}^{d-1}))$ \; \label{alg_mpc:line7}
    $\share{\bm{F}[2\rho+2^{h+1}]}^A,\share{\bm{F}[2\rho+2^{h+1}-1]}^A \leftarrow \mathtt{SelectShare}(\share{2}^A, \share{1}^A, {\prod}_{\rm EQ}(\share{\bm{F}[\rho+2^h-1]}^A, 2))$ \; \label{alg_mpc:line8}
    % $\share{\bm{F}[2\rho+2^{h+1}]}^A = \share{\bm{F}[2\rho+2^{h+1}-1]}^A$ \; \label{alg_mpc:line9}
    $\share{\bm{F}[\rho+2^h-1]}^A \leftarrow \mathtt{SelectShare}(\share{0}^A, \share{1}^A, {\prod}_{\rm EQ}(\share{{\Psi}_{r_{0}}}^A, 0)~{\wedge}$ $~{\prod}_{\rm EQ}(\share{{\Psi}_{r_{1}}}^A, 0)~{\wedge}~{\prod}_{\rm EQ}(\share{\bm{F}[\rho+2^h-1]}^A, 1))$ \; \label{alg_mpc:line10}
    
}
\end{algorithm}
\noindent\textbf{MPC-based HC.}
${\prod}_{\rm OHC}^{mpc}$ uses the relatively expensive MPC protocol $\mathtt{Division}$ to compute $\frac{P}{Q}$, and then uses $\mathtt{Maxpool}$ to compare the division results and select the best feature with the smallest gini index.
We use the $\mathtt{Division}$ and $\mathtt{Maxpool}$ proposed in Falcon~\cite{wagh2020falcon}.

The Gini index value is typically a floating-point value, while our MPC protocols operate over discrete domains like rings and finite fields. 
To address this, we use a fixed-point encoding strategy~\cite{tan2021cryptgpu,ng2021gforce,wagh2020falcon}. 
In details, a real value $x{\in}\RR$ is converted into an integer ${\lfloor}x{\cdot}2^{\tau}{\rceil}$ (\ie the nearest integer to $x{\cdot}2^{\tau}$) with $\tau$ bits of precision. 
The value of $\tau$ affects both the performance and accuracy. A smaller fixed-point precision enables the shares over a $32$-bit ring instead of $64$-bit, thus reducing communication and computation costs. 
However, lower precision may cause Gini index values to appear identical, potentially leading to incorrect feature selection.
Therefore, selecting the proper $\tau$ requires careful experimentation.

${\prod}_{\rm OHC}^{mpc}$ is given in Algorithm~\ref{alg_mpc}. The main steps are division (Line~\ref{alg_mpc:line4}) and selecting the best feature (Line~\ref{alg_mpc:line5}-\ref{alg_mpc:line7}).

% Although MPC-based HC keeps all computations on the three GPUs, 
As shown in CryptGPU~\cite{tan2021cryptgpu}, $\mathtt{Division}$ is GPU-unfriendly. %, making it the performance bottleneck of \sys.
%On the one hand, $\mathtt{Division}$ is less GPU-friendly; On the other hand, our secure DT training requires computing gini index at each leaf or dummy node.
How to make $\mathtt{Division}$ GPU-friendly and avoid division in DT training remain open challenges. 
To enhance performance, we alternatively propose to outsource the HC phase to TEEs, \ie TEE-based HC: ${\prod}_{\rm OHC}^{tee}$. 
We discuss the trade-offs between ${\prod}_{\rm OHC}^{mpc}$ and ${\prod}_{\rm OHC}^{tee}$ in Section~\ref{sec:exp training}. 

\normalem
\begin{algorithm}[!ht]
\footnotesize
\DontPrintSemicolon 
\caption{Oblivious Heuristic Computation via TEE, ${\prod}_{\rm OHC}^{tee}$}
\label{alg4}
\SetKwInOut{Input}{Input}\SetKwInOut{Output}{Output}

\Input{$\{\share{\bm{C_n}}^A\}_{n{\in}[0,n_h-1]}$, $\{\share{\bm{\gamma_n}}^A\}_{n{\in}[0,n_h-1]}$, $\share{\bm{F}}^A$, $h$}
\Output{Split decisions array $\share{\bm{SD}}^A$, updated $\{\share{\bm{\gamma_n}}^A\}_{n{\in}[0,n_h-1]}$, $\share{\bm{F}}^A$}

%Three GPUs send their respective shares of the input arrays and public level $h$ to the SGX.\; \label{alg4:line1}
$\{\bm{C_n}\}_{n{\in}[0,n_h-1]}=\mathtt{Reconstruct}^A(\{\share{\bm{C_n}}^A\}_{n{\in}[0,n_h-1]})$\; \label{alg4:line1}
$\{\bm{\gamma_n}\}_{n{\in}[0,n_h-1]}=\mathtt{Reconstruct}^A(\{\share{\bm{\gamma_n}}^A\}_{n{\in}[0,n_h-1]})$\; \label{alg4:line2}
$\bm{F}=\mathtt{Reconstruct}^A(\share{\bm{F}}^A)$\; \label{alg4:line3}
\For{$\rho=\{0,\cdots,n_h-1\}$}{ \label{alg4:line4}
    \For{$i=\{0,\cdots,d-2\}$}{ \label{alg4:line5}
        $\bm{gini}[i] = \overline{G}(\bm{C}_{\rho}, s_i)$\; \label{alg4:line6}
        %N_D,, $\bm{D_t}$ represents the data samples at $\rho$-th leaf.
    }
    % $\bm{SD}[t] = \mathtt{oselect}(\mathtt{oless}(\bm{gini}, \bm{\gamma_t}), s_{r}, \bm{F}[t+2^h-1]==1)$, where $r \stackrel{\$}{\leftarrow}\{0, d-2\}$ \; \label{alg4:line7}
    \nonl \textcolor{gray}{// select the indices of the best features with $\mathtt{oless}$}\;
    $\bm{SD}[\rho] = \mathtt{oless}(\bm{gini}, \bm{\gamma_{\rho}})$ \; \label{alg4:line7}
    $\mathtt{oassign}(\bm{\gamma_{\rho}}, 0, \bm{SD}[\rho])$ \; \label{alg4:line8}
    % $\bm{F}[t+2^{h+1}-1] = \mathtt{oselect}(2, 1, \bm{F}[t+2^h-1]==2)$; $\bm{F}[t+2^{h+1}] = \bm{F}[t+2^{h+1}-1]$ \; \label{alg4:line9}
    $\bm{F}[2\rho+2^{h+1}], \bm{F}[2\rho+2^{h+1}-1] = \mathtt{oselect}(2, 1, \bm{F}[\rho+2^h-1]==2)$ \; \label{alg4:line9}
    $CheckLeaf$ indicates if $\rho$-th node contains only one label \; \label{alg4:line10}
    %and only one of $\bm{C_t}[1][0]+\bm{C_t}[1][1]$ and $\bm{C_t}[2][0]+\bm{C_t}[2][1]$ is not 0 (\ie only contain a label value)\; \label{alg4:line9}
    $\bm{F}[\rho+2^h-1] = \mathtt{oselect}(0, 1, (\neg CheckLeaf \wedge (\bm{F}[\rho+2^h-1]==1)))$ \; \label{alg4:line11}
} 
$\share{\bm{SD}}^A=\mathtt{Share}^A(\bm{SD})$ \; \label{alg4:line12}
$\share{\bm{F}}^A=\mathtt{Share}^A(\bm{F})$ \; \label{alg4:line13}
$\share{\{\bm{\gamma_n}\}_{n{\in}[0,n_h-1]}}^A=\mathtt{Share}^A(\{\bm{\gamma}_n\}_{n{\in}[0,n_h-1]})$ \; \label{alg4:line14}
TEE sends the respective shares of the outputs to the three GPUs \;\label{alg4:line15}
\end{algorithm}
\noindent\textbf{TEE-based HC.}
Algorithm~\ref{alg4} shows ${\prod}_{\rm OHC}^{tee}$.
Basically, parties send the shares of all current $\{\bm{C_n},\bm{\gamma_n}\}_{n{\in}[0,n_h-1]}$ and $\bm{F}$ to the co-located TEE enclave. 
The enclave reconstructs the inputs, generates/updates the outputs using similar steps as ${\prod}_{\rm OHC}^{mpc}$, and finally sends the respective shares to each CSP. 
% ($n_h=2^h$ given current level $h$)
To hide access patterns, the enclave uses the oblivious primitives (as used in~\cite{ohrimenko2016oblivious,law2020secure,poddar2020visor}): $\mathtt{oless}(\bm{vec},\bm{cond})$, return the index of the minimum value in $\bm{vec}$. $\bm{cond}$ indicates the compared elements; $\mathtt{oselect}(a, b, cond)$, select $a$ or $b$ based on condition $cond$; $\mathtt{oassign}(\bm{vec}, a, idx)$: assign the value of $\bm{vec}$ at index $idx$ with $a$.
These oblivious primitives operate on registers, making the instructions immune to memory-access-based pattern leakage once the operands are loaded into registers.

\normalem
\begin{algorithm}[htbp]
\footnotesize
\DontPrintSemicolon 
\caption{Oblivious Node Split, ${\prod}_{\rm ONS}$}
\label{alg5}
\SetKwInOut{Input}{Input}\SetKwInOut{Output}{Output}

\Input{$\share{\bm{SD}}^A$, $\share{\bm{F}}^A$, $\{\share{\bm{C_n}}^A\}_{n{\in}[0,n_h-1]}$ at level $h$}
\Output{Updated $\share{\bm{T}}^A$, new arrays $\{\share{\bm{C_n}}^A\}_{n{\in}[0,N_c-1]}$ at level $h+1$ where $N_c=2n_h$}

% \For{$\rho=\{0,\cdots,n_h-1\}$} { 
%     $\share{\bm{isLeaf}}^B=\{{\prod}_{\rm EQ}(\share{\bm{F}[t+2^h-1]}^A,1)\}^2$ \; \label{alg4:line1}
%     Parties initialize arrays $\bm{C}$ for two child leaves of $\rho$-th node: $\share{\bm{C'_t}}^A=\{\share{\bm{C'_j}}^A\}_{j{\in}[0,1]}$ where $\share{\bm{C'_j}}^A=\share{\bm{C_t}}^A$\;
%     $\share{\bm{C'_t}}^A {\leftarrow} \mathtt{SelectShare}(\{\share{0}^A\}^2,\share{\bm{C'_t}}^A,\share{\bm{isLeaf}}^B)$\;
% }
\For{$\rho=\{0,\cdots,n_h-1\}$} { \label{alg5:line1}
    $\share{isInternal}^B \leftarrow {\prod}_{\rm EQ}(\share{\bm{F}[\rho+2^h-1]}^A,0)$ \; \label{alg5:line2}
    % \For{array $\share{\bm{C}}^A$ of each child leaf of $\rho$-th node} { 
    %     $\share{\bm{C}}^A {\leftarrow} \mathtt{SelectShare}(\share{0}^A,\share{\bm{C_t}}^A,\share{\bm{isLeaf}}^B)$\;
    % }
    \For{each $\share{\bm{C_i}}^A$ in $\share{\bm{C'_{\rho}}}^A=\{\share{\bm{C_i}}^A\}_{i{\in}[0,1]}\}$} { \label{alg5:line3}
        $\share{\bm{C_i}}^A {\leftarrow} \mathtt{SelectShare}(\share{\bm{C_{\rho}}}^A,\share{0}^A,\share{isInternal}^B)$\; \label{alg5:line4}
    }
}
Concatenate all $\{\share{\bm{C'_{\rho}}}^A\}_{\rho \in [0,n_h-1]}$ into $\{\share{\bm{C_n}}^A\}_{n{\in}[0,N_c-1]}$\; \label{alg5:line5}
Replace the values of $\share{\bm{T}}^A$ at level $h$ with $\share{\bm{SD}}^A$\; \label{alg5:line6}
\end{algorithm}

% \noindent\textbf{Oblivious Node Split.}
\subsubsection{\textbf{Oblivious Node Split}}
\label{sec:protocol:ons}
The next step is to convert each leaf into an internal node using the best feature and generate new leaves at the next level (Algorithm~\ref{alg5}).
Recall that \sys builds the DT in layer-wise. 
\sys updates $\bm{C}$ of all new leaves (Line~\ref{alg5:line1}-\ref{alg5:line4}) and $\bm{T}$ (Line~\ref{alg5:line6}) in parallel.
Array $\bm{C}$ of a new leaf inherits from its parent node if the parent node is a dummy or leaf node; otherwise, set all values to zeros (Line~\ref{alg5:line4}).
Notably, $\bm{\gamma}$ of each new leaf inherits from its parent node no matter which type the parent node is.

% \noindent\textbf{Oblivious Decision Tree Training.}
\subsubsection{\textbf{Oblivious Decision Tree Training}}
\label{sec:protocol:odtt}

Our training procedure is composed of the aforementioned protocols.
As a result of our level-wise tree construction, \sys always builds a full binary tree with depth $H$. 
$H$ could be: 
(1) a pre-defined public depth;
(2) a depth where all nodes at this level are leaf or dummy nodes;
(3) a depth which is the number of features $d$ (including the label).
Such design has a trade-off between privacy and performance, since we may add many dummy nodes as the tree grows, especially in the sparse DT.
% We discuss this with evaluation results in Section~\ref{sec:exp training}.

\subsection{\textbf{Oblivious DT Inference}}
\label{sec:protocol:odti}

\normalem
\begin{algorithm}[!ht]
\footnotesize
\DontPrintSemicolon 
\caption{Oblivious DT Inference, ${\prod}_{\rm ODTI}$}
\label{alg8}
\SetKwInOut{Input}{Input}\SetKwInOut{Output}{Output}

\Input{Unlabelled input $\share{\bm{I}}^A$ of length $N_I$, $\share{\bm{T}}^A$, tree depth $H$}
\Output{Inference result $\share{\bm{R}}^A$ of length $N_I$}

Initialize $\share{\bm{Tidx}}^A$, $\share{\bm{Tval}}^A$, $\share{\bm{Dval}}^A$ of length $N_I$, $n_{iter}=0$ \; \label{alg8:line1}
\While{$n_{iter} < H$}{ \label{alg8:line2}
    $\share{\bm{Tval}}^A \leftarrow {\prod}_{\rm OAA}(\share{\bm{T}}^A, \share{\bm{Tidx}}^A)$ \; \label{alg8:line3}
    \For{$i = \{0,\cdots,N_I-1\}$}{ \label{alg8:line4}
        $\share{\bm{Dval}[i]}^A \leftarrow {\prod}_{\rm OAA}(\share{\bm{I}[i]}^A, \share{\bm{Tval}}^A)$ \; \label{alg8:line5}
    }
    Compute $\share{\bm{Tidx}}^A = 2\share{\bm{Tidx}}^A + \share{\bm{Dval}}^A + 1$ \; \label{alg8:line6}
    Compute $n_{iter}++$ \; \label{alg8:line7}
}
% \sout{$\share{R}^A={\prod}_{\rm OAA}(\share{Tval}^A, \share{L}^A)$.}\; \label{alg9:line7}
$\share{\bm{R}}^A=\share{\bm{Tval}}^A$ \;\label{alg8:line8}
\end{algorithm}

${\prod}_{\rm ODTI}$ protocol is shown in Algorithm~\ref{alg8}.
The main operation is to obliviously access $\bm{T}$ and $\bm{I}$ (Line~\ref{alg8:line3}-\ref{alg8:line5}). 
The values of $\bm{T}$ at the last level $H$ are the inference results (Line~\ref{alg8:line8}). 
Notably, due to the benefits of GPU-friendly property in \sys, we can process a large number of concurrent queries. 

\section{Performance Evaluation}
\label{sec:exp}

We build \sys on top of Piranha~\cite{watson2022piranha}, a platform for accelerating secret sharing-based MPC protocols using GPUs. 
Piranha provides some state-of-the-art secret-sharing MPC protocols that allow us to easily leverage GPU.
Particularly, some of these protocols such as $\mathtt{SelectShare}$, $\mathtt{Division}$ and $\mathtt{Maxpool}$~\cite{wagh2020falcon} are used in \sys.
Following with Piranha, our main focus is on evaluating \sys's performance in the data-dependent online phase, as offline generation of data-independent components, such as edaBits~\cite{makri2021mathsf}, can be easily parallelized independently from a specific computation.

In this section, we evaluate the performance of \sys training and inference. 
We also compare the performance of \sys with previous CPU-based solutions and two baselines.

\subsection{Experiment Setup}
\noindent\textbf{Testbed.}
We evaluate the prototype on a server equipped with 3 NVIDIA Tesla V100 GPUs, each of which has 32 GB of GPU memory. 
The server runs Ubuntu 20.04 with kernel version 5.4.0 and has 10-core Intel Xeon Silver 4210 CPUs (2.20GHz) and 125GB of memory.
We implement \sys in C++ and marshal GPU operations via CUDA 11.6. 
Since the GPU server does not support Intel SGX, we use a desktop with SGX support to test SGX-based protocol ${\prod}_{\rm OHC}^{tee}$. The desktop contains 8 Intel i9-9900 3.1GHZ cores and 32GB of memory (${\thicksim}93$ MB EPC memory) and runs OpenEnclave 0.16.0~\cite{microsoftOE}.

\noindent\textbf{Network Latency.}
As done in previous work~\cite{abspoel2020secure,adams2022privacy,hoogh2014practical}, we test all the protocols on the same machine. Piranha provides the functionality to emulate the network connection and measure network latency based on the communication overhead and the bandwidth. 
Our test uses a local area network (LAN) environment with a bandwidth of 10 Gbps and ping time of 0.07 ms (same setting as~\cite{abspoel2020secure,hoogh2014practical}) \textcolor{black}{to simulate three independent CSPs}.
Communication between GPUs is bridged via CPU in Piranha~\cite{watson2022piranha} and \sys inherits this property. 
The time for transferring data between GPU and CPU on the same machine is negligible in \sys.
All the following results are average over at least 10 runs.
% \noindent\textbf{Baselines.}

\begin{table}[!ht]
	\centering
        \footnotesize
	\caption{Datasets}
	\label{tab:dataset}
        \resizebox{!}{0.7cm}{%
    	\begin{tabular}{ccccc}
    		\toprule
    		\textbf{Dataset} & \textbf{\#Features} & \textbf{\#Labels}  & \textbf{\#Samples}  & \textbf{Tree Depth} \\	\hline
    		SPECT     & $22$  & $2$  & 267    & 6  \\	%\hline
    		KRKPA7    & $35$  & $2$  & 3,196  & 9  \\
    		Adult     & $14$  & $2$  & 48,842 & 5  \\	\bottomrule
    	\end{tabular}
        }
\end{table}

\subsection{Performance of \sys Training}
\label{sec:exp training}

\begin{table}[htbp]
\centering
\caption{Performance of three steps (in seconds).}
\label{tab:time_step}
\resizebox{!}{3.5cm}{%
    \begin{tabular}{ccccccc}
    \toprule
    \multirow{2}{*}{} & \multirow{2}{*}{$\prod_{\text{OL}}$} & \multicolumn{2}{c}{$\prod_{\text{OHC}}$} & \multirow{2}{*}{$\prod_{\text{ONS}}$} & \multicolumn{2}{c}{Total} \\
    \cline{3-4} \cline{6-7}
     &  & \textbf{tee} & \textbf{mpc} &  & \textbf{tee} & \textbf{mpc} \\
    \midrule
    \multicolumn{7}{l}{(a) $(d,H)=(8, 6)$, vary data samples}  \\
    \midrule
    $1 \times 10^4$ & 1.63 & 0.01 & 1.35 & 0.07 & 1.71 & 3.05 \\
    $2 \times 10^4$ & 3.04 & 0.01 & 1.35 & 0.07 & 3.12 & 4.46 \\
    $3 \times 10^4$ & 4.45 & 0.01 & 1.36 & 0.07 & 4.53 & 5.88 \\
    $4 \times 10^4$ & 5.95 & 0.02 & 1.36 & 0.07 & 6.04 & 7.38 \\
    $5 \times 10^4$ & 7.42 & 0.02 & 1.38 & 0.07 & 7.51 & 8.87 \\
    \midrule
    \multicolumn{7}{l}{(b) $(N_D,H)=(50000, 5)$, vary features}  \\
    \midrule
    
    4  & 2.99  & 0.01 & 1.06 & 0.05 & 3.05 & 4.10 \\
    8  & 3.83  & 0.01 & 1.08 & 0.05 & 3.89 & 4.96 \\
    16 & 5.37  & 0.01 & 1.11 & 0.05 & 5.43 & 6.53 \\
    32 & 8.16  & 0.02 & 1.15 & 0.05 & 8.24 & 9.37 \\
    64 & 14.26 & 0.02 & 1.21 & 0.06 & 14.34 & 15.53 \\
    \midrule
    \multicolumn{7}{l}{(c) $(N_D,d)=(50000, 8)$, vary depth}  \\
    \midrule
    4 & 2.01  & 0.01 & 0.85 & 0.04 & 2.06 & 2.90 \\
    5 & 3.78  & 0.01 & 1.10 & 0.05 & 3.83 & 4.93 \\
    6 & 7.53  & 0.01 & 1.40 & 0.07 & 7.59 & 8.68 \\
    7 & 15.46 & 0.01 & 1.65 & 0.08 & 15.53 & 16.62 \\
    8 & 33.30 & 0.03 & 1.96 & 0.10 & 33.39 & 35.36 \\
    \bottomrule
    \end{tabular}
}
\end{table}

We first evaluate the performance of the protocols used in DT training:  ${\prod}_{\rm OL}$, ${\prod}_{\rm OHC}^{mpc}$, ${\prod}_{\rm OHC}^{tee}$, and ${\prod}_{\rm ONS}$.
For this test, we use a synthetic dataset which allows us to flexibly change the number of samples, features, and depths so as to better show the performance of \sys under different conditions.
The results are shown in Table.~\ref{tab:time_step}.   

Table.~\ref{tab:time_step} shows that the learning phase is the most expensive step. For training $5{\times}10^4$ data samples with depth 8, \sys spends 33.24 seconds, taking up about 98\% of the overall training runtime on average.
The main reason is that this step uses ${\prod}_{\rm OAA}$ and $\mathtt{SelectShare}$ with a large input size (\,  e.g. $N_D$) multiple times, incurring high communication costs.

\noindent\textbf{Heuristic Computation: ${\prod}_{\rm OHC}^{mpc}$ vs. ${\prod}_{\rm OHC}^{tee}$.} For the heuristic computation phase, as discussed in Section~\ref{sec:protocol:ohc}, in ${\prod}_{\rm OHC}^{mpc}$,
the value of $\tau$ affects both the performance of ${\prod}_{\rm OHC}^{mpc}$ and the model accuracy. 
From our experiment results, we observe that when setting $\tau=10$ and working over the ring $\ZZ_{2^{32}}$, the model accuracy trained using ${\prod}_{\rm OHC}^{mpc}$ is acceptable. 
Thus, for evaluating ${\prod}_{\rm OHC}^{mpc}$, we set $\tau=10$.

For the accuracy of \sys when setting $\tau=10$, here we present the results tested with 3 UCI datasets that are widely used in the literature: SPECT Heart (SPECT) dataset, Chess (King-Rook vs. King-Pawn) (KRKPA7) dataset, Adult dataset (see details in Table~\ref{tab:dataset}). 
We split the dataset into $80\%$ for training and $20\%$ for inference. 
\textcolor{black}{We adjust various hyper-parameters (\eg depth and minimum samples per leaf~\cite{quinlan2014c4}) to obtain the best accuracy for comparisons\footnote{Note that techniques for improving accuracy such as DT pruning~\cite{quinlan2014c4} are not considered in our tests}.}
The DT model trained by \sys using ${\prod}_{\rm OHC}^{tee}$ achieves the same inference accuracy as the plaintext model.
However, the DT model trained with ${\prod}_{\rm OHC}^{mpc}$ experiences an accuracy drop of $1\%-4\%$.

The findings reveal that, even with sufficient precision, accuracy loss can still occur since the feature assigned to each node might not be the best one. This is primarily attributed to the following factors:
(i) With relatively large precision and datasets, the numerator and denominator, represented by $P$ and $Q$ in Equation~\ref{eq_gini_reduction}, become sizable. To prevent overflow, we use the scaling function from Piranha~\cite{watson2022piranha} to scale $P$ and $Q$, introducing potential accuracy loss.
(ii) The $\mathtt{Division}$ operation employs an approximation technique (\ie Taylor expansion~\cite{tan2021cryptgpu,wagh2020falcon,watson2022piranha}), which may introduce a certain degree of accuracy loss.
(iii) The Gini index values of different features can be very close, leading to incorrect selections when compared under MPC. 
Additionally, controlling precision proves challenging as it necessitates varying precisions at each node and remains unpredictable.
Note that in different datasets, this may not always result in accuracy drop, as features with similar Gini indices may explain the model well.

When $\tau=10$, Table~\ref{tab:time_step} shows that ${\prod}_{\rm OHC}^{tee}$ outperforms ${\prod}_{\rm OHC}^{mpc}$ by $136{\times}$.  
Compared with ${\prod}_{\rm OL}$, the time taken by ${\prod}_{\rm OHC}^{tee}$ is almost negligible. 
Therefore, for training $5\times 10^4$ data samples, training with ${\prod}_{\rm OHC}^{tee}$ is only ${\thicksim}1.3 \times$ faster.
%than the one using ${\prod}_{\rm OHC}^{mpc}$. 

\noindent\textbf{Remark.}
${\prod}_{\rm OHC}^{tee}$ surpasses ${\prod}_{\rm OHC}^{mpc}$ in both accuracy and performance. Nevertheless, employing ${\prod}_{\rm OHC}^{tee}$ also comes with certain limitations. For example, utilizing TEE such as Intel SGX necessitates dependence on Intel's security mechanisms and incurs the overhead of mitigating side-channel attacks, \eg access-pattern-based attacks (see the comparison of \sys and TEE-only baseline in Fig.~\ref{fig:sgx_only_gtree}). Overall, a trade-off exists between these two approaches. While ${\prod}_{\rm OHC}^{tee}$ offers better performance, ${\prod}_{\rm OHC}^{mpc}$ has the potential for future optimization through more advanced MPC protocols.

\begin{figure*}[!ht]
  \centering
    \begin{subfigure}[t]{0.6\columnwidth}
      \centering   
      \includegraphics[width=\linewidth]{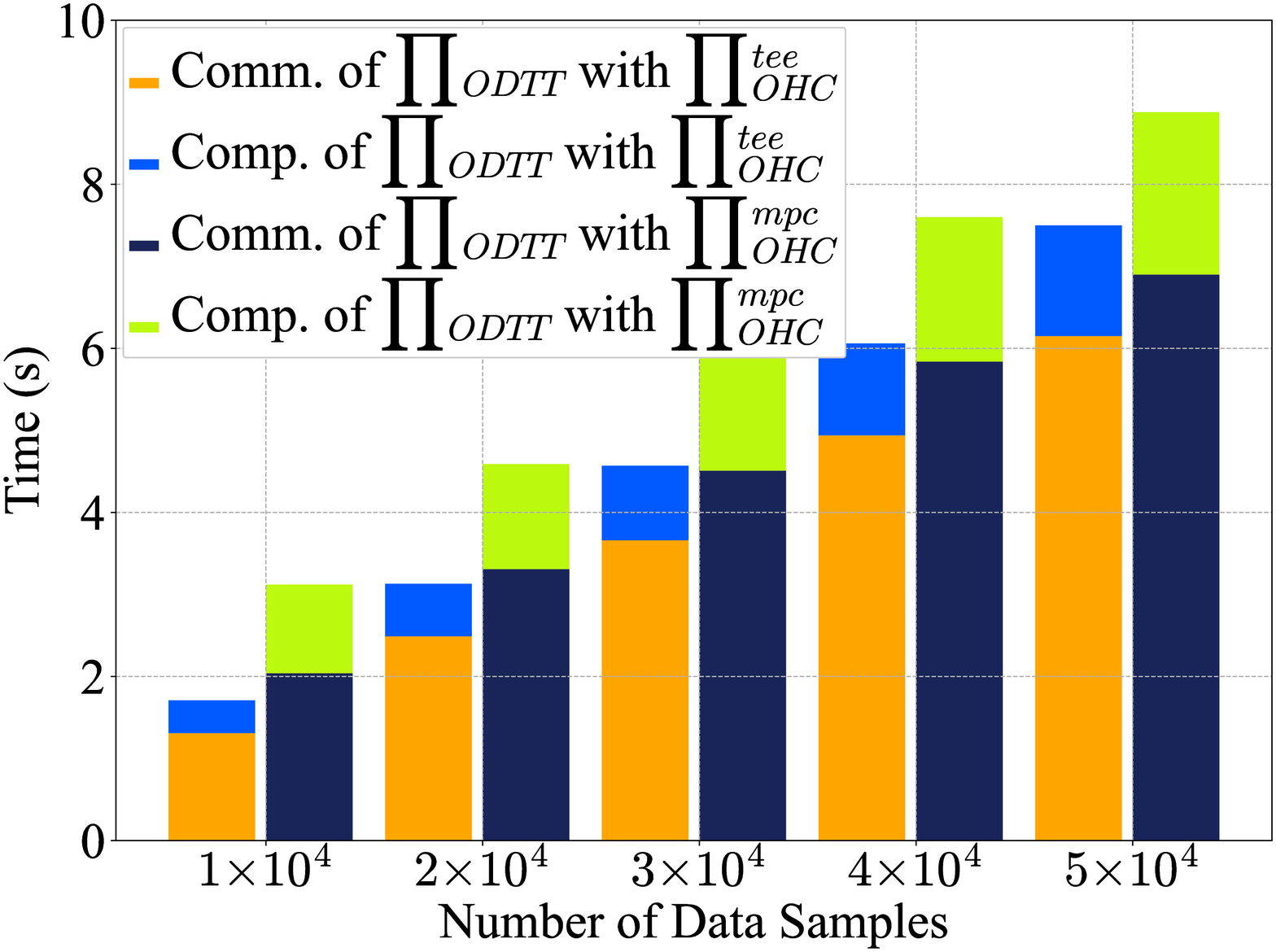}
        \caption{With $8$ features and 6 levels}
        \label{fig:time_comm_comp:sub1}
    \end{subfigure}   %      \hfill  % 
    \centering
    \begin{subfigure}[t]{0.6\columnwidth}
      \centering   
      \includegraphics[width=\linewidth]{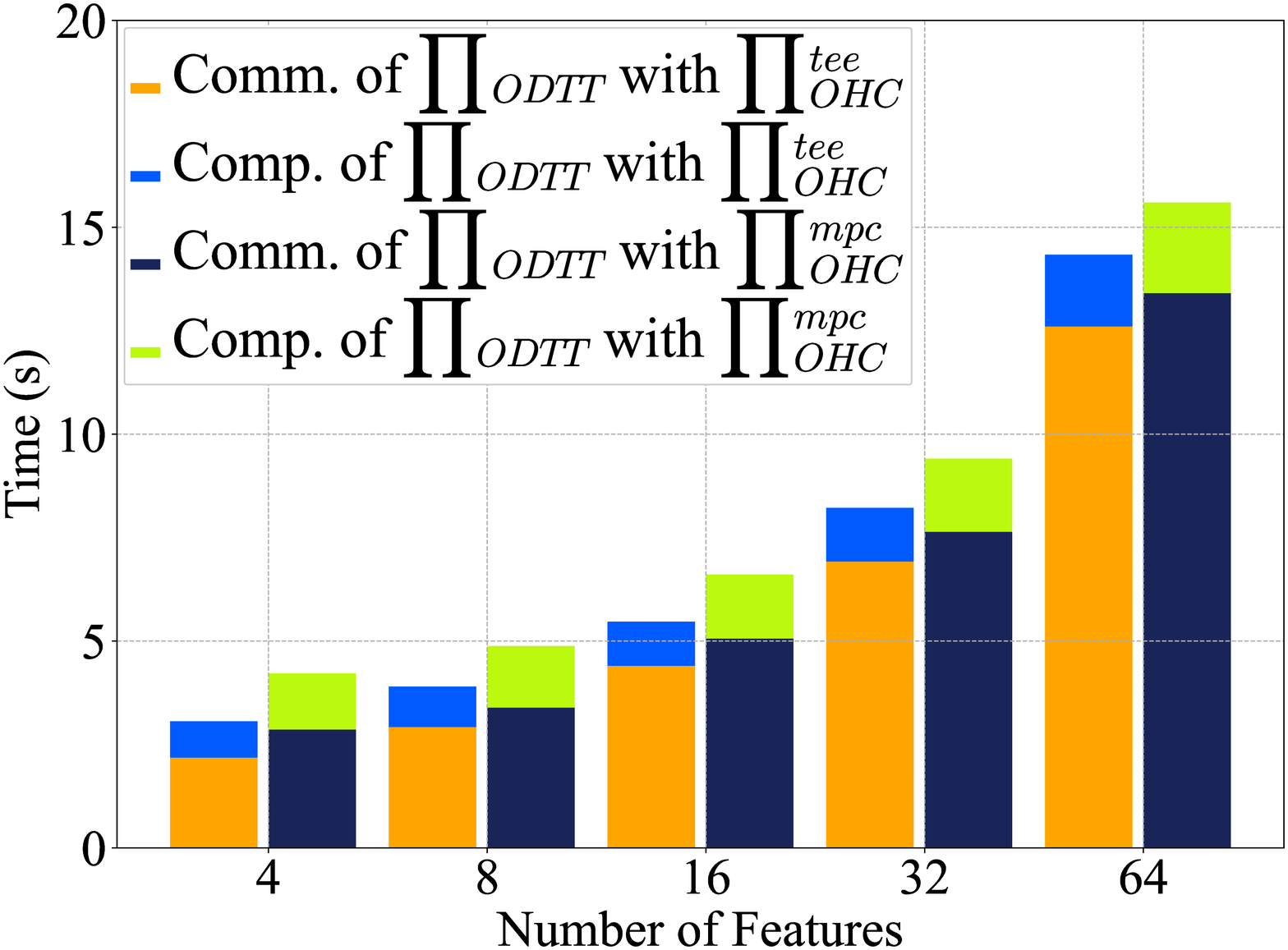}
        \caption{With $50,000$ samples and 5 levels}
        \label{fig:time_comm_comp:sub2}
    \end{subfigure}
    \begin{subfigure}[t]{0.6\columnwidth}
      \centering   
      \includegraphics[width=\linewidth]{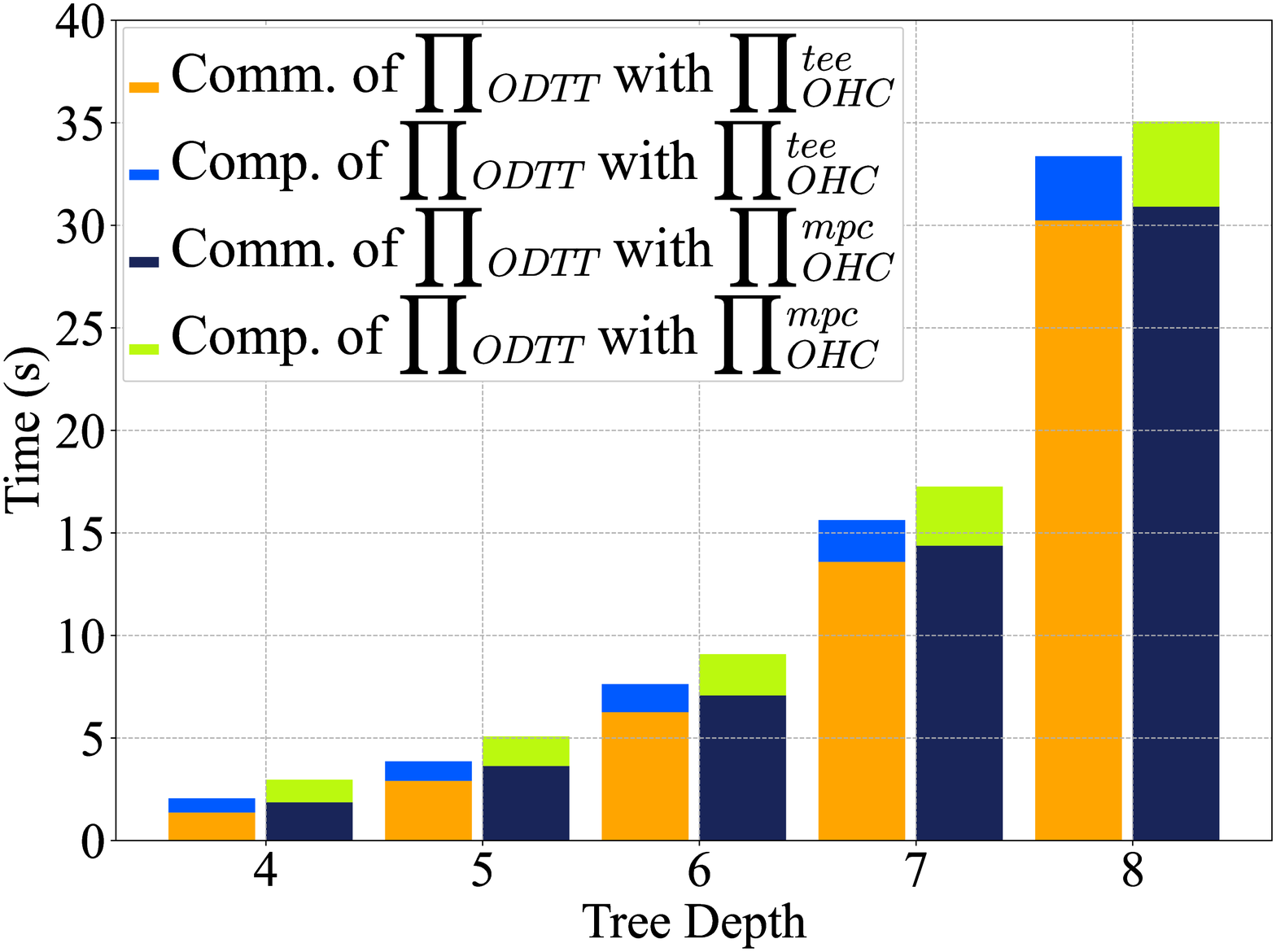}
        \caption{With $8$ features and $50,000$ samples}
        \label{fig:time_comm_comp:sub3}
    \end{subfigure}
\caption{
\label{fig:time_comm_comp}
Computation (Comp.) and Communication (Comm.) time of training. 
}
\end{figure*}

\begin{table}[!ht]
    \renewcommand\arraystretch{1.3}
    \centering
    \footnotesize
    \caption{Comm. cost (MB) for different depths (8 features)}
    \label{tab:comm_mb_depth}
    \resizebox{!}{1.4cm}{%
    \begin{tabular}{lccccc}
        \toprule
        \textbf{\#Depths} & \textbf{4} & \textbf{5} & \textbf{6} & \textbf{7} & \textbf{8} \\ \hline
        ${\prod}_{\rm OL}$                  &     460.1       &      936.6      &       1924.9      &     3998.2        &      8363.4       \\ %\hline
        ${\prod}_{\rm OHC}^{tee}$                 &     0.004       &      0.008      &      0.017       &      0.034       &      0.070       \\ %\hline
        ${\prod}_{\rm OHC}^{mpc}$                 &      0.258      &       0.553     &      1.14       &       2.32      &      4.68       \\ %\hline
        ${\prod}_{\rm ONS}$                 &      0.009      &     0.020       &      0.041       &       0.084      &       0.170      \\ \hline
        \makecell{${\prod}_{\rm ODTT}$ with ${{\prod}_{\rm OHC}^{tee}}$}           &    460.1        &      936.6      &      1925       &      3998.4       &   8363.6          \\ %\hline
        \makecell{${\prod}_{\rm ODTT}$ with ${{\prod}_{\rm OHC}^{mpc}}$}           &      460.4      &      937.1      &      1926.1       &      4000.6       &  8368.3  \\
        \bottomrule
    \end{tabular}
    }
\end{table}

% \subsubsection{Computation and Communication Cost.}
\noindent\textbf{Communication Overhead.}
% The main issue of MPC is the heavy communication overhead. 
We measure the communication overhead during training in \sys with different tree depths ($5\times 10^4$ data samples) in Table~\ref{tab:comm_mb_depth} (overhead varies more with tree depth). 
Not surprisingly, the learning phase incurs the highest communication cost due to the large volume of data processed. 
Additionally, ${\prod}_{\rm OHC}^{tee}$ is more communication-efficient than ${\prod}_{\rm OHC}^{mpc}$. 

In Fig.~\ref{fig:time_comm_comp}, we split the training time into the time spent on communicating and computing. 
With our LAN setting, about $69\%-90\%$ (also varies with the network bandwidth) overhead of \sys training is the communication, which is similar to CryptGPU~\cite{tan2021cryptgpu} and Piranha~\cite{watson2022piranha}. We will optimize the communication overhead in our future work (see Section~\ref{sec:conclusion}). 
It is worth noting that even with such communication overhead, \sys is still highly more efficient than CPU-based and TEE-based solutions (see Section~\ref{subsec:comp}).

\subsection{Performance of \sys Inference}
\label{sec:inf_bench}

\begin{figure}[!ht]
  \centering
    \begin{subfigure}{0.49\columnwidth}
      \centering   
      \includegraphics[width=\textwidth]{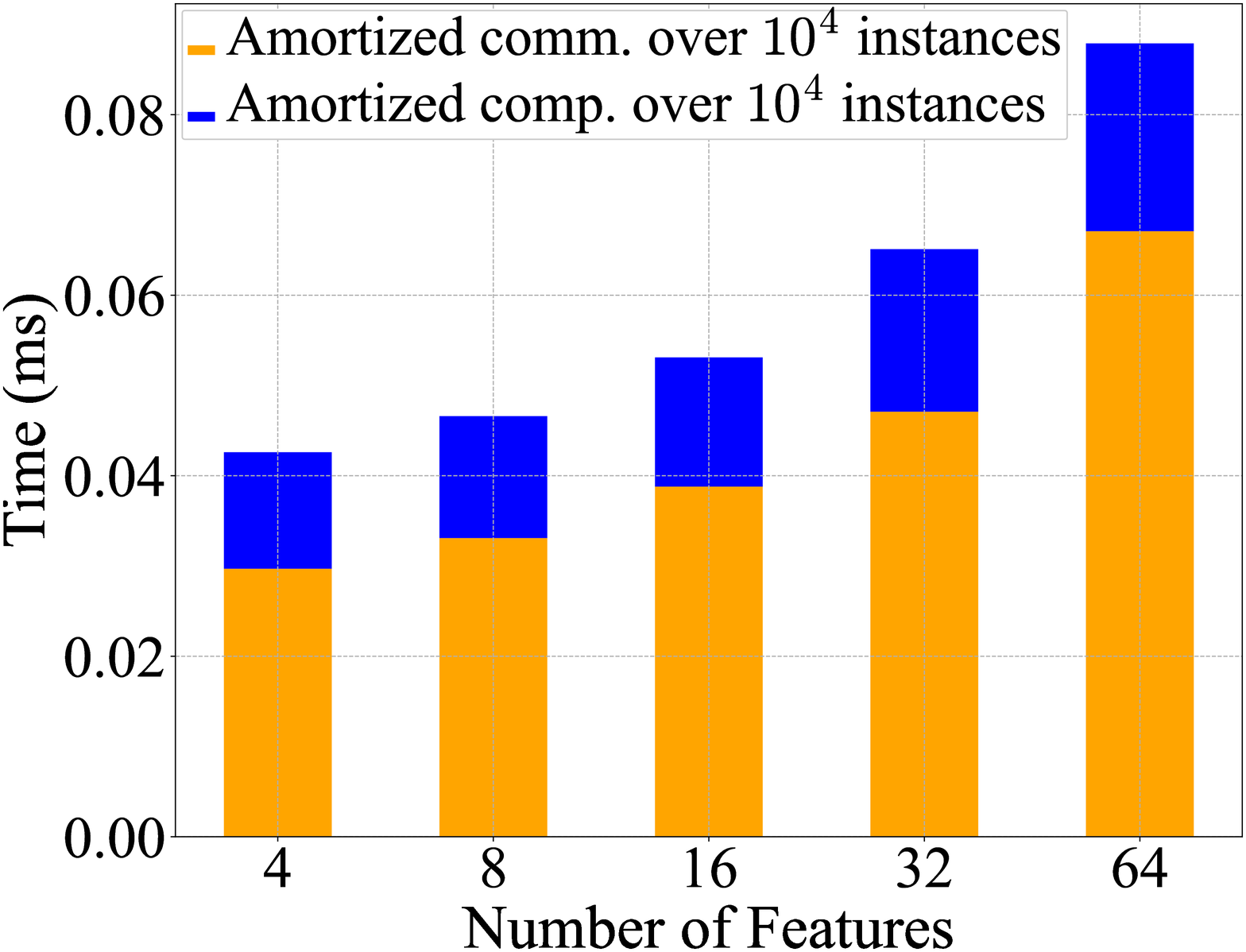}
        \caption{With 5 tree depth}
        \label{fig:inf_comm_comp:sub1}
    \end{subfigure}   %      \hfill  % 
    \begin{subfigure}{0.49\columnwidth}
      \centering   
      \includegraphics[width=\textwidth]{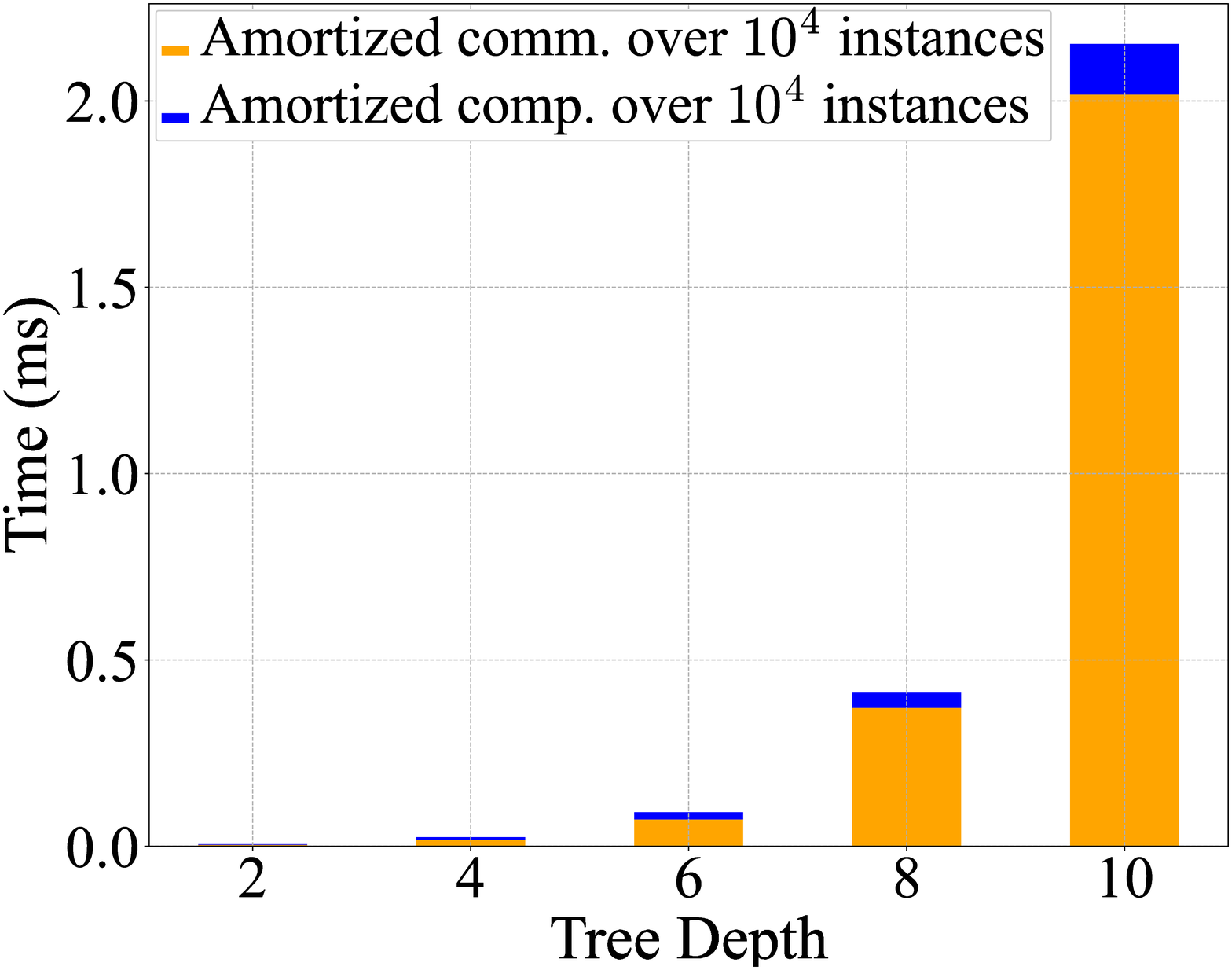}
        \caption{With 10 features}
        \label{fig:inf_comm_comp:sub2}
    \end{subfigure}
\caption{
\label{fig:inf_comm_comp}
Comp. and comm. time of inference. \textnormal{The results are amortized over $10^4$ instances, \ie infer $10^4$ inferences at once.}  
}
\end{figure}

To test inference, we input $10^4$ instances and measure the amortized time, which is broken down into computation and communication times in Fig.~\ref{fig:inf_comm_comp}. 
Note that Fig.~\ref{fig:inf_comm_comp} only illustrates the costs among CSPs, excluding the latency between QU and CSP. 
The main overhead of inference is also the communication, which takes about $69\%-90\%$. 

Comparing Fig.~\ref{fig:inf_comm_comp:sub1} and Fig.~\ref{fig:inf_comm_comp:sub2}, we can see tree depth has a greater impact on inference performance than the number of features, and when the tree depth $\geq 10$, \sys's inference time increases sharply. 
Algorithm~\ref{alg8} shows that DT inference primarily relies on multiple invocations of ${\prod}_{\rm OAA}$, which depends on the tree depth $H$ ($N_I$ inputs are inferred in parallel). 
Thus, for inference, \sys is better suited for trees with small or medium depths. 
For deeper trees, using random forests may be more effective than a single DT.

\subsection{Performance Comparisons with Others}
\label{subsec:comp}
\begin{table}[!ht]
	\centering
	\footnotesize
	\caption{Time of training with UCI datasets (in seconds)}
	\label{tab:performance_real}
        \resizebox{!}{0.6cm}{%
	\begin{tabular}{cccc}
		\toprule
		\textbf{Scheme} & \textbf{SPECT}  & \textbf{KRKPA7}  & \textbf{Adult} \\
		\hline
		$\mathtt{SID3T}$ & 3.55 & 6.45 (not secure) & 89.07  \\
		\textbf{\sys} & 0.31 & 8.63 & 4.32 \\
		\bottomrule
	\end{tabular}
 }
\end{table}
\noindent\textbf{Comparisons with Prior Work.}
We also compare the performance of \sys against Hoogh \etal~\cite{hoogh2014practical} and Liu \etal~\cite{liu2020towards}.
To the best of our knowledge, these are the only private DT works that demonstrated the ability to train the dataset with categorical features in a relatively efficient manner.
Several recent works~\cite{abspoel2020secure,adams2022privacy,hamada2021efficient} focus on designing specified MPC protocols (\eg sorting) to process continuous features.
\textcolor{black}{We recognize that both \sys and this line of research can handle both types of features by integrating additional MPC-based data-processing protocols (\eg MPC-based discretization~\cite{adams2022privacy}). In this experiment, we primarily compare our approach with schemes that also concentrate on categorical features.}

\begin{figure*}[!ht]
  \centering
    \begin{subfigure}[t]{0.65\columnwidth}
      \centering   
      \includegraphics[width=\linewidth]{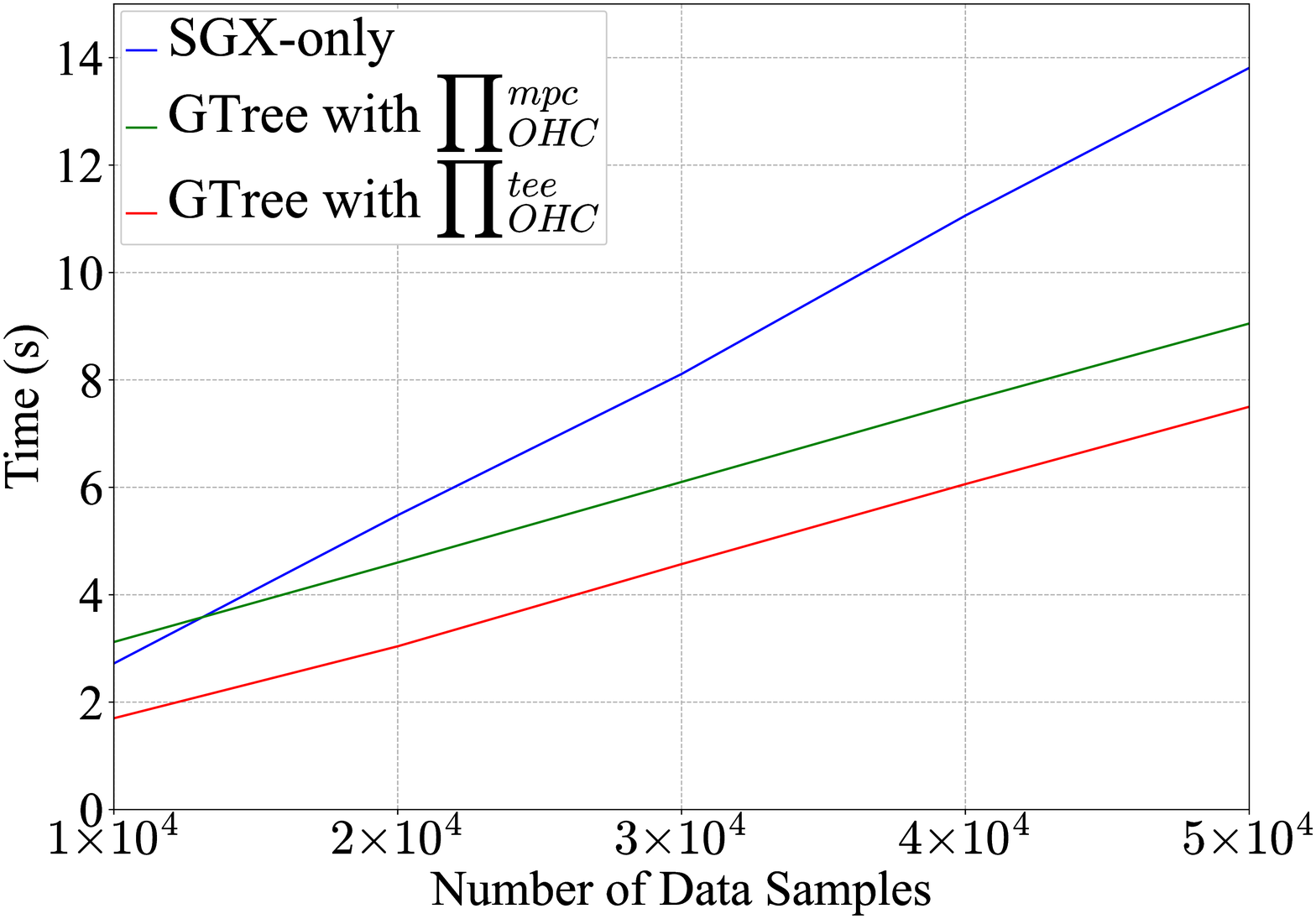}
        \caption{\small Training, 8 features and 6 levels}
        \label{fig:sgx_only_gtree:sub1}
    \end{subfigure}   %      \hfill  % 
    \centering
    \begin{subfigure}[t]{0.65\columnwidth}
      \centering   
      \includegraphics[width=\linewidth]{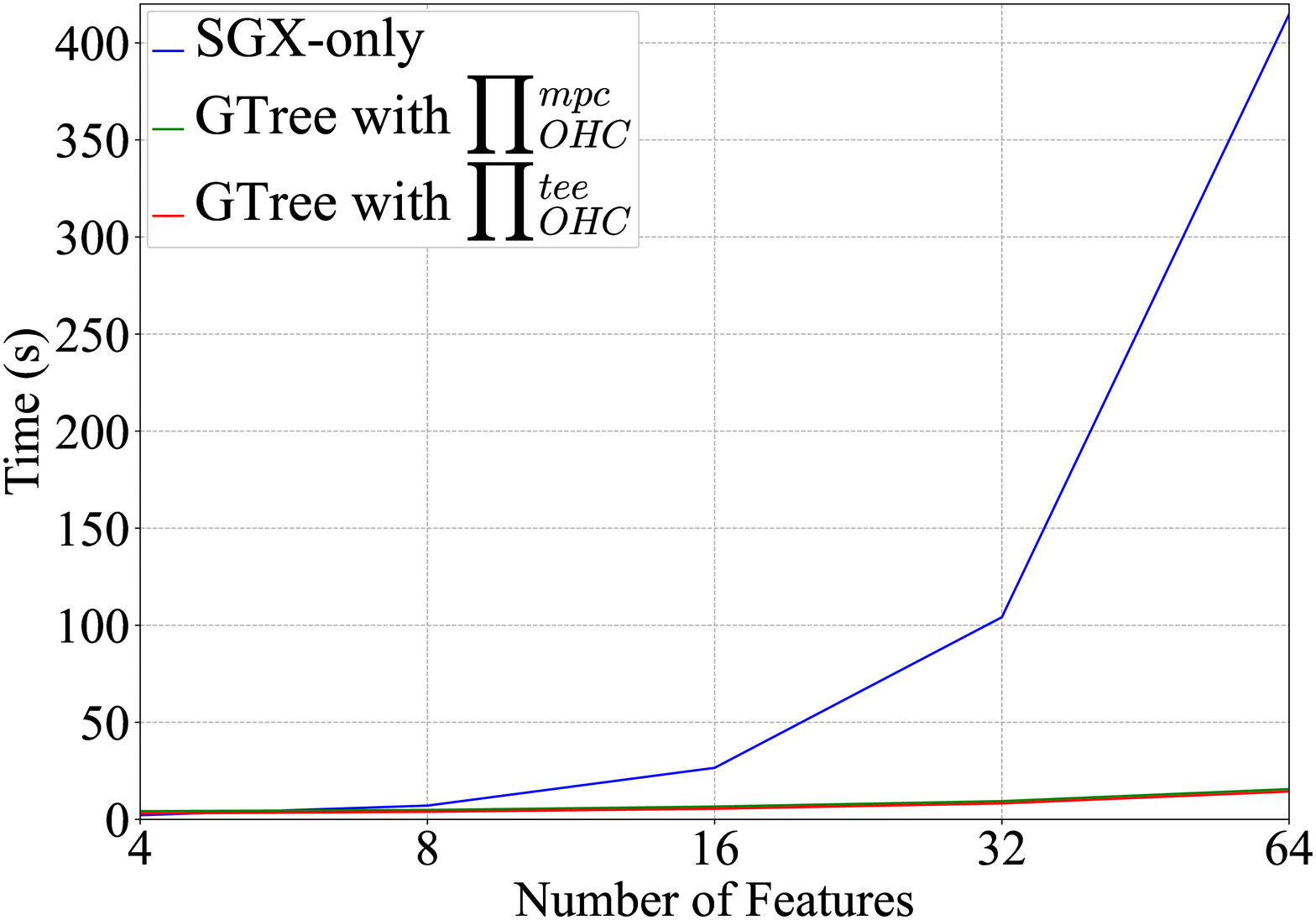}
        \caption{\small Training, 50,000 data and 5 levels}
        \label{fig:sgx_only_gtree:sub2}
    \end{subfigure}
    \begin{subfigure}[t]{0.65\columnwidth}
      \centering   
      \includegraphics[width=\linewidth]{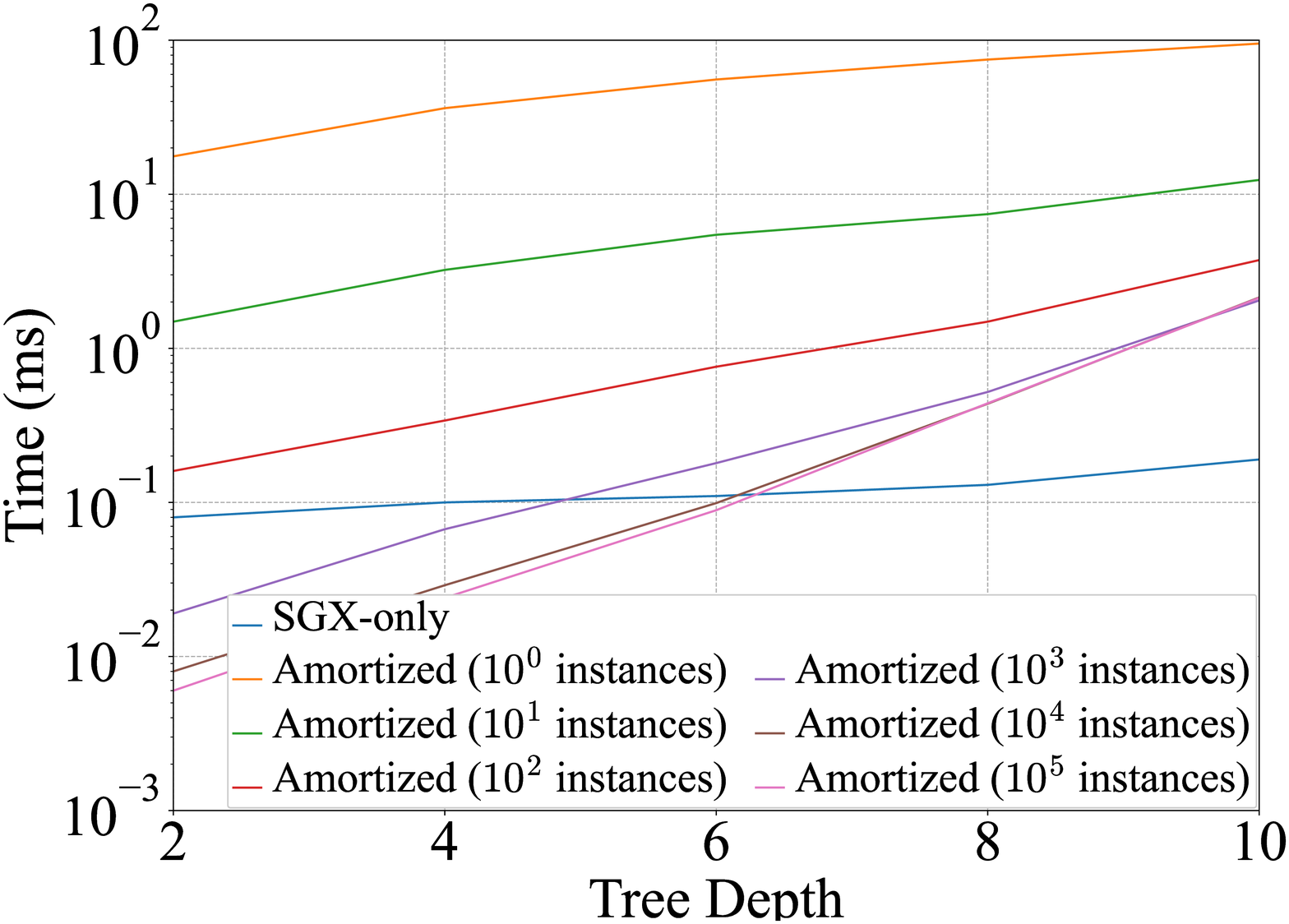}
        \caption{\small Inference, 10 features}
        \label{fig:sgx_only_gtree:sub3}
    \end{subfigure}
\caption{
\label{fig:sgx_only_gtree}
\small Performance of $\mathtt{SGX~only}$ and \sys for training and inference. Note that $y$-axis in Fig.~\ref{fig:sgx_only_gtree:sub3} is in logarithm scale.
}
\end{figure*}

As described in most recent works~\cite{abspoel2020secure,adams2022privacy}, Hoogh \etal~\cite{hoogh2014practical} is still the most state-of-the-art approach to process categorical features.
They provide three protocols with different levels of security.
We implement their protocol with the highest security level based on their latest MPC framework, referring to this baseline as $\mathtt{SID3T}$.
However, $\mathtt{SID3T}$ is less secure because it does not protect the tree shape, which is typically one of the most expensive aspects of secure DT training.

\textcolor{black}{In Table~\ref{tab:performance_real}, we first evaluated the performance of DT training using the three commonly used UCI datasets from other DT schemes~\cite{hoogh2014practical,adams2022privacy,abspoel2020secure}.}
% Datasets from UCI repository: 
Table~\ref{tab:dataset} shows the dataset details.
The DT depths trained over these datasets are 6, 9, and 5. 
$\mathtt{SID3T}$ is tested on the CPU of the same machine used in \sys.
The results are shown in Table~\ref{tab:performance_real}.
For SPECT and Adult, \sys outperforms $\mathtt{SID3T}$ by ${\thicksim}11{\times}$ and ${\thicksim}21{\times}$, respectively.
When training with KRKPA7, \sys's performance is slightly worse than $\mathtt{SID3T}$.
This is because \sys trains a full binary tree up to the depth to protect the tree structure.
Nevertheless, $\mathtt{SID3T}$ avoids this expensive part at the cost of security loss.
Overall, \sys demonstrates superior performance while providing stronger security.

We compare \sys with Liu \etal~\cite{liu2020towards} based on their reported results since their code is not publicly available at the time of writing. 
The largest dataset they used is the Tic-tac-toe dataset with 958 data samples and 9 features.
Training with this dataset, \sys takes 0.43 seconds and 39.61 MB communication cost, yielding ${\thicksim}3,112{\times}$ and ${\thicksim}49{\times}$ improvements, respectively.
Notably, Liu \etal~\cite{liu2020towards} also do not protect the tree structure.

\begin{table}[htbp]
	\centering
        \footnotesize
	\caption{Inference time (in milliseconds)}
	\label{tab:inf_cmp_others}
        \resizebox{!}{0.68cm}{%
	\begin{tabular}{cccccccc}
		\toprule
		\textbf{Dataset} & \textbf{d}  & \textbf{H} & \textbf{GGH~\cite{kiss2019sok}} & \textbf{JZL~\cite{ji2022uc}} & \textbf{\sys} & \textbf{Sp.GGH} & \textbf{Sp.JZL} \\	\hline
		Wine      & 7   & 5  &  14  &  8  &  0.05  & $280{\times}$  & $160{\times}$ \\	
            Breast    & 12  & 7  &  24  &  17  &  0.19  & $126{\times}$  & $90{\times}$ \\	
            Digits    & 47  & 15 &  115 &  34  &  102.7  & $1.1{\times}$  & $0.33{\times}$ \\	\bottomrule
	\end{tabular}
        }
    \begin{tablenotes}
        % \centering
        \footnotesize
        \item \textbf{d}, \textbf{H}: the number of features, depth. \textbf{Sp.GGH} and \textbf{Sp.JZL} represent the speedup that \sys achieves, compared to GGH and JZL, respectively.
    \end{tablenotes}
\end{table}

For DT inference, Kiss \etal~\cite{kiss2019sok} compare most of existing 2PC private inference schemes and identified \textit{GGH} as the most efficient in terms of online runtime. 
Additionally, JZL~\cite{ji2022uc}'s approach stands out as the most efficient 3PC scheme to date. 
In Table~\ref{tab:inf_cmp_others}, we provide an online runtime comparison of \sys with GGH and JZL~\cite{ji2022uc} based on the runtimes reported in their respective papers. 
\textcolor{black}{We perform the preprocessing for discretization on these datasets for \sys.}
The inference times of \sys for the three datasets in Table~\ref{tab:inf_cmp_others} are amortized over $10^4$, $10^4$, and $10^3$ instances (due to GPU memory constraints), respectively. 
Notably, \sys incurs a low offline cost since only edaBits~\cite{makri2021mathsf} need to be generated, compared to GGH and JZL~\cite{ji2022uc}. 
Basically, \sys achieves remarkable performance gains for small and medium-sized trees. 
This is because when the tree depth is small, the benefits of high parallelism in \sys fully offset the increased communication overhead.
However, as the tree depth increases, the communication overhead in \sys does not scale linearly, leading to suboptimal performance in deep trees. 
This calls for future optimizations for deeper trees (see discussions in Section~\ref{sec:conclusion}).

\noindent\textbf{Comparisons with TEE-only Solution.}
As mentioned, protecting data processes with TEE is another research line in the literature. Although GPU itself cannot provide protection, due to its powerful parallelism capability, GPU is a better option. 
Another reason is that existing TEEs (\eg Intel SGX~\cite{costan2016intel}) suffer from side-channel attacks. 
To protect the model effectively from attacks, the tasks performed within TEEs should be oblivious, which is expensive. 
Here, we use Intel SGX as an example and implement secure DT that processes everything inside an SGX enclave using oblivious primitives, \eg $\mathtt{oassign}$ and $\mathtt{oaccess}$~\cite{ohrimenko2016oblivious,law2020secure,poddar2020visor} (denote it as \textit{SGX-only}).
The evaluation results are given in Fig.~\ref{fig:sgx_only_gtree}. 

For training, Fig.~\ref{fig:sgx_only_gtree:sub1} and ~\ref{fig:sgx_only_gtree:sub2} show that \sys outperforms the SGX-only solution almost for all cases.  
Note that we omit the case for different tree depths since it has a similar trend as in Fig.~\ref{fig:sgx_only_gtree:sub2}.
For inference (Fig.~\ref{fig:sgx_only_gtree:sub3}), SGX-only is much more efficient than \sys when inferring only one instance. 
However, the performance of \sys improves as the number of instances increases, thanks to better utilization of GPU parallelism. When tree depth reaches 10, GPU resources are fully utilized with concurrent inference of $10^3$ instances. 
For a tree depth of 6, \sys outperforms the SGX-only solution when inferring more than $10^4$ instances concurrently, with the advantage becoming more pronounced at shallower tree depths. Additionally, \sys's performance will improve further with more powerful GPUs.

\noindent\textbf{Comparisons with Plaintext Training.}
As done in CryptGPU~\cite{tan2021cryptgpu}, we report the comparison results between $\mathtt{Insecure}$ baseline and \sys, where $\mathtt{Insecure}$ baseline trains the same datasets in plaintext with one GPU. 
For CNN training, CryptGPU still adds roughly $2000\times$ overhead compared with the insecure training. 
For DT training with \sys, the gap is less than $1673\times$. 
However, the gap between them is still large, which highlights the need to develop more GPU-friendly cryptographic primitives in the future.

\section{Security Analysis}
\label{sec:security}

We prove the security of our protocols using the real-world/ideal-world simulation paradigm~\cite{goldreich2019play}.
We define the entities as follows: $\mathcal{A}$: A \textit{static semi-honest} probabilistic polynomial time (PPT) real-world adversary; $\mathcal{S}$: The corresponding ideal-world adversary (simulator); $\mathcal{F}$: The ideal functionality. The adversary $\mathcal{A}$ can corrupt at most one party at the beginning and follows the protocol honestly.
In the real world, the parties interact with $\mathcal{A}$ and the environment $\mathcal{Z}$, executing the protocol as specified.
In the ideal world, parties send their inputs to a trusted party that computes the functionality accurately.
For every real-world adversary $\mathcal{A}$, there exists a simulator $\mathcal{S}$ in the ideal world such that no environment $\mathcal{Z}$ can distinguish between the real and ideal worlds. 
This ensures that any information $\mathcal{A}$ can extract in the real world can also be extracted by $\mathcal{S}$ in the ideal world.
We use multiple sub-protocols in the sequential model and employ the hybrid model for security proofs. The hybrid model simplifies proof analysis by replacing sub-protocols with their corresponding ideal functionalities.
A protocol that invokes a functionality $\mathcal{F}$ is said to be in an ``$\mathcal{F}$-hybrid model".

We define the respective simulators $\mathcal{F}_{\rm EQ}$, $\mathcal{F}_{\rm OAA}$, $\mathcal{F}_{\rm OL}$, $\mathcal{F}_{\rm OHC}^{mpc}$, $\mathcal{F}_{\rm OHC}^{sgx}$, $\mathcal{F}_{\rm ONS}$ and $\mathcal{F}_{\rm ODTI}$ for protocols ${\prod}_{\rm EQ}$, ${\prod}_{\rm OAA}$, ${\prod}_{\rm OL}$, ${\prod}_{\rm OHC}^{mpc}$, ${\prod}_{\rm OHC}^{sgx}$, ${\prod}_{\rm ONS}$ and ${\prod}_{\rm ODTI}$.
The ideal functionalities for $\mathcal{F}_{\rm Mult}$ and $\mathcal{F}_{\rm Reconst}$ are identical to prior works~\cite{wagh2020falcon}.
We prove security using the standard indistinguishability argument. 
By setting up hybrid interactions where their ideal functionalities replace sub-protocols, these interactions can be simulated as indistinguishable from real ones.

The following analysis and theorems demonstrate this indistinguishability.

\noindent \textbf{Security of ${\prod}_{\rm EQ}$.} In ${\prod}_{\rm EQ}$, we mainly modify from Rabbit~\cite{makri2021mathsf}, where the involved computations are all local.
Therefore, the simulator for $\mathcal{F}_{\rm EQ}$ follows easily from the original protocol, which has been proved secure.

\noindent \textbf{Security of ${\prod}_{\rm OAA}$.} We capture the security of ${\prod}_{\rm OAA}$ as Theorem~\ref{theorem_oaa} and give the detailed proof as follows.

\begin{theorem}
\label{theorem_oaa}
${\prod}_{\rm OAA}$ securely realizes $\mathcal{F}_{\rm OAA}$, in the presence of one semi-honest party in the ($\mathcal{F}_{\rm Mult}$, $\mathcal{F}_{\rm EQ}$, $\mathcal{F}_{\rm SelectShare}$)-hybrid model.
\end{theorem}

\begin{proof}
The simulation follows easily from the protocol and the hybrid argument.
The simulator runs the first iteration of the loop (Step~\ref{alg2:line1}) and in the process extracts the inputs. 
Then it proceeds to complete all the iterations of the loop.
The simulator for $\mathcal{F}_{\rm EQ}$ can be used to simulate the transcripts from Step~\ref{alg2:line2}.
The simulator for $\mathcal{F}_{\rm SelectShare}$ follows from the protocol in Falcon~\cite{wagh2020falcon}.
Steps~\ref{alg2:line3} and ~\ref{alg2:line5} are all local operations and do not need simulation. 
Therefore, ${\prod}_{\rm OAA}$is secure in the ($\mathcal{F}_{\rm Mult}$, $\mathcal{F}_{\rm EQ}$, $\mathcal{F}_{\rm SelectShare}$)-hybrid model.
% Finally, if the protocol aborts at any time in the internal run, then the simulator sends $\mathrm{Abort}$ to $\mathcal{F}_{OAA}$;
% otherwise, it inputs the extracted shares of $\share{\bm{W}}^A$ and $\share{\bm{U}}^A$ to to $\mathcal{F}_{OAA}$ and the honest parties receive their outputs.
\end{proof}

\noindent \textbf{Security of ${\prod}_{\rm OL}$.} We capture the security of ${\prod}_{\rm OL}$ as Theorem~\ref{theorem_ol} and give the detailed proof as follows.
\begin{theorem}
\label{theorem_ol}
${\prod}_{\rm OL}$ securely realizes $\mathcal{F}_{\rm OL}$, in the presence of one semi-honest party in the ($\mathcal{F}_{\rm Mult}$, $\mathcal{F}_{\rm EQ}$, $\mathcal{F}_{\rm OAA}$, $\mathcal{F}_{\rm SelectShare}$)-hybrid model.
\end{theorem}

\begin{proof}
The simulator for $\mathcal{F}_{\rm Mult}$ can be used to simulate the transcripts from Steps~\ref{alg3:line10},~\ref{alg3:line16}-~\ref{alg3:line19}.
${\prod}_{\rm OL}$ are sequential combinations of local computations (Steps~\ref{alg3:line5},~\ref{alg3:line6},~\ref{alg3:line11},~\ref{alg3:line14},~\ref{alg3:line15},~\ref{alg3:line21}) and invocations of $\mathcal{F}_{\rm Mult}$, $\mathcal{F}_{\rm EQ}$ (Steps~\ref{alg3:line7},~\ref{alg3:line9}), $\mathcal{F}_{\rm OAA}$ (Steps~\ref{alg3:line2},~\ref{alg3:line4}) and $\mathcal{F}_{\rm SelectShare}$ (Step~\ref{alg3:line20}).
The simulation follows directly from composing the simulators.
\end{proof}

\noindent \textbf{Security of ${\prod}_{\rm OHC}^{mpc}$.} 
Protocol ${\prod}_{\rm OHC}^{mpc}$ is composed of protocols such as $\mathtt{Division}$ and $\mathtt{Maxpool}$ from Falcon~\cite{wagh2020falcon}.
The security statement and proof of ${\prod}_{\rm OHC}^{mpc}$ are as follows:

\begin{theorem}
\label{theorem_ohc_mpc}
${\prod}_{\rm OHC}^{mpc}$ securely realizes $\mathcal{F}_{OHC}^{mpc}$, in the presence of one semi-honest party in the ($\mathcal{F}_{\rm Mult}$, $\mathcal{F}_{\rm EQ}$, $\mathcal{F}_{\rm SelectShare}$, $\mathcal{F}_{\rm Maxpool}$, $\mathcal{F}_{\rm Division}$)-hybrid model.
\end{theorem}

\begin{proof}
The simulators for $\mathcal{F}_{\rm Maxpool}$, $\mathcal{F}_{\rm Division}$ and $\mathcal{F}_{\rm SelectShare}$ follow from Falcon~\cite{wagh2020falcon}.
${\prod}_{\rm OHC}^{mpc}$ is sequential combinations of local computations and the corresponding simulators.
\end{proof}

% define ideal functionality for SGX part

% \normalem
% \begin{algorithm}[htbp]
% \DontPrintSemicolon
% \caption{$\mathcal{F}_{OHC}^{sgx}$: Ideal functionality for ${\prod}_{\rm OHC}^{sgx}$}
% \label{alg_sgx_proof}
% \SetKwInOut{Input}{Input}\SetKwInOut{Output}{Output}

% \Input{The functionality receives inputs $\{\share{\bm{C_n}}^A\}_{n{\in}[0,n_h-1]}$, $\{\share{\bm{\gamma_n}}^A\}_{n{\in}[0,n_h-1]}$, $\share{\bm{F}}^A$ and $h$}
% \Output{Compute the following}

% Reconstruct $\{\bm{C_n}\}_{n{\in}[0,n_h-1]}$, $\{\bm{\gamma_n}\}_{n{\in}[0,n_h-1]}$, $\bm{F}$ \;\label{alg_sgx_proof:line1}
% \end{algorithm}

% \begin{theorem}
% \label{theorem_ohc_sgx}
% ${\prod}_{\rm OHC}^{sgx}$ securely realizes $\mathcal{F}_{OHC}^{sgx}$ 
% \end{theorem}
\noindent \textbf{Security of ${\prod}_{\rm OHC}^{sgx}$.}
We explain the ideal functionality for ${\prod}_{\rm OHC}^{sgx}$ according to the proof of $\mathcal{F}_{\rm attest}$ in CRYPTFLOW~\cite{kumar2020cryptflow}.
$\mathcal{F}_{\rm OHC}^{sgx}$ is realized as follows:
Intel SGX guarantees confidentiality by creating a secure enclave where code and data can be securely executed and stored.
Initially, when SGX receives a command for computing ${\prod}_{\rm OHC}$, it performs a remote attestation with the party.
Once attested, the data transmitted between the party and the enclave will be encrypted with a secret key $sk$.
Upon receiving input from the parties, the enclave executes the code and produces secret-shared outputs encrypted under $sk$.
When running inside the enclave, even if an adversary compromises the host system, he cannot learn the data from the enclave.

\noindent \textbf{Security of ${\prod}_{\rm ONS}$.} We capture the security of ${\prod}_{\rm ONS}$ as Theorem~\ref{theorem_ons} and give the detailed proof as follows.

\begin{theorem}
\label{theorem_ons}
${\prod}_{\rm ONS}$ securely realizes $\mathcal{F}_{\rm ONS}$, in the presence of one semi-honest party in the ($\mathcal{F}_{\rm EQ}$, $\mathcal{F}_{\rm SelectShare}$)-hybrid model.
\end{theorem}

\begin{proof}
Similar to the proof of Theorem~\ref{theorem_ol}, simulation works by sequentially composing the simulators for $\mathcal{F}_{\rm EQ}$ and $\mathcal{F}_{\rm SelectShare}$.
\end{proof}

%${\prod}_{\rm ODTT}$ and
\noindent \textbf{Security of ${\prod}_{\rm ODTI}$.} We capture the security of ${\prod}_{\rm ODTI}$ in Theorem~\ref{theorem_odti}, and give their formal proofs as follows:

% \begin{theorem}
% \label{theorem_odtt}
% ${\prod}_{\rm ODTT}$ securely realizes $\mathcal{F}_{\rm ODTT}$, in the presence of one semi-honest party in the ($\mathcal{F}_{\rm OL}$, $\mathcal{F}_{\rm OHC}^{mpc}$/$\mathcal{F}_{\rm OHC}^{sgx}$, $\mathcal{F}_{\rm ONS}$)-hybrid model.
% \end{theorem}

% \begin{proof}
% % Simulation works by sequentially composing those simulators.
% Simulation is done using the hybrid argument.
% The protocol simply composes $\mathcal{F}_{\rm OL}$, $\mathcal{F}_{\rm OHC}^{mpc}$/$\mathcal{F}_{\rm OHC}^{sgx}$ and $\mathcal{F}_{\rm ONS}$ and hence is simulated using the corresponding simulators.
% \end{proof}

\begin{theorem}
\label{theorem_odti}
${\prod}_{\rm ODTI}$ securely realizes $\mathcal{F}_{\rm ODTI}$, in the presence of one semi-honest party in the $\mathcal{F}_{\rm OAA}$-hybrid model.
\end{theorem}

\begin{proof}
The simulation follows easily from the hybrid argument.
Simulation works by sequentially composing $\mathcal{F}_{\rm OAA}$ and hence is simulated using the corresponding simulator.
\end{proof}

\section{Related Work}
\label{sec:related_work}
In this section, we review existing privacy-preserving approaches for general ML algorithms and explore the most recent works that leverage GPU acceleration. We finally survey the work for privacy-preserving DT training and inference.

\subsection{Private machine learning using GPU}
Privacy-preserving ML has received considerable attention over recent years.
Recent works operate in different models such as deep learning~\cite{wagh2020falcon,ohrimenko2016oblivious} and tree-based models~\cite{abspoel2020secure,adams2022privacy,liu2020towards,law2020secure,lindell2000privacy,emekcci2007privacy,hoogh2014practical,tueno2019private,tai2017privacy}.
%gilad2016cryptonets,juvekar2018gazelle,agrawal2019quotient,bost2015machine,wu2016privately,de2017efficient,agrawal2019quotient,
These works rely on different privacy-preserving techniques, such as MPC~\cite{wagh2020falcon,abspoel2020secure,adams2022privacy,lindell2000privacy,emekcci2007privacy,hoogh2014practical,tueno2019private}, HE~\cite{liu2020towards,tai2017privacy}, TEE~\cite{ohrimenko2016oblivious,law2020secure}.
% bost2015machine,wu2016privately,de2017efficient,huang2022cheetah,
However, most of them demonstrate a CPU-only implementation and focus mainly on improving the performance of their specified protocols. 
There is an urgent need to further improve practical performance when deploying these in real-world applications.

Recently, few works explored GPU-based MPC in private deep learning.
% Two earliest works~\cite{frederiksen2014faster,husted2013gpu} leverage GPU to speed up some basic operations in MPC.
% Delphi~\cite{mishra2020delphi} explores the use of GPU for improving the performance of linear layers.
CryptGPU~\cite{tan2021cryptgpu} shows the benefits of GPU acceleration for both training and inference on top of the CrypTen framework.
GForce~\cite{ng2021gforce} proposes an online/offline GPU/CPU design for inference with GPU-friendly protocols.
Visor~\cite{poddar2020visor} is the first to combine the CPU TEE with GPU TEE in video analytics, yet it is closed-source.
Moreover, it requires hardware modification which would adversely affect compatibility.
% and the more recent NVIDIA H100 GPU \cite{h100gpu2022}
% Most of them are not easy-to-use due to their poor applicability and high cost.
Overall, all of these works focus on deep learning including massive GPU-friendly computations (\ie convolutions and matrix multiplications).
\sys is the first to support secure DT training and inference on the GPU.
%Therefore, there is a urgent call to combine GPU acceleration with privacy-preserving DT training and inference.

\subsection{Privacy-preserving Decision Tree}
% \notewang{To-be-done, add some DP work, explain weak security model (esorics comments, but unnecessary).}
\noindent\textbf{Inference.}
Most of the existing works~\cite{bost2015machine,wu2016privately,de2017efficient,tueno2019private,tai2017privacy,bai2022scalable,ohrimenko2016oblivious} focus mainly on DT inference. 
Given a pre-trained DT model, they ensure that the QU (or say client) learns as little as possible about the model and the CSP learns nothing about the queries.
SGX-based approaches~\cite{ohrimenko2016oblivious} have demonstrated that they are orders of magnitude faster than
cryptography-based approaches~\cite{bost2015machine,wu2016privately,de2017efficient,tueno2019private,tai2017privacy}.
However, all of these works only consider the scenario of a single query from the QU.
When encountering a large number of concurrent queries, the performance of the above approaches degrades significantly.
\sys is superior due to GPU parallelism.
%\sout{For example, \sys yields an amortized improvement of up to ${\thicksim}2,446{\times}$ per instance. }

\noindent\textbf{Training.}
Privacy-preserving training~\cite{abspoel2020secure,adams2022privacy,liu2020towards,law2020secure,lindell2000privacy,emekcci2007privacy,hoogh2014practical} is naturally more difficult than inference.
This is because the training phase involves more information leakage and more complex functions.
Since Lindell \etal~\cite{lindell2000privacy} initialize the study of privacy-preserving data mining, there has been a lot of research~\cite{emekcci2007privacy,hoogh2014practical} on DT training.
However, most of them focus only on data privacy while do not consider the model (\eg tree structures).
In more recent work, Liu \etal~\cite{liu2020towards} design the protocols for both training and inference on categorical data by leveraging additive HE and secret sharing.
However, they do not protect the patterns of building the tree and take over 20 minutes to train a tree of depth 7 with 958 samples.
More recently, schemes~\cite{abspoel2020secure,adams2022privacy} propose new training algorithms for processing continuous data using MPC.
However, Abspoel \etal~\cite{abspoel2020secure} have to train the random forest instead of a single DT when processing a large dataset.
Adams \etal~\cite{adams2022privacy} train other models (\eg random forest and extra-trees classifier) to bypass some expensive computations in the original DT. 
All in all, they do not explore the use of GPU in private DT domain.
% More importantly, the protocols proposed in~\cite{abspoel2020secure,adams2022privacy} can be integrated into \sys for processing continuous data with GPU acceleration.
\section{Conclusion and Future Work}
\label{sec:conclusion}

In this work, we propose \sys, the first framework that combines privacy-preserving decision tree training and inference with GPU acceleration. 
\sys achieves a stronger security guarantee than previous work, where both the access pattern and the tree shape are also hidden from adversaries, in addition to the data samples and tree nodes. 
\sys is designed in a GPU-friendly manner that can take full advantage of GPU parallelism. 
Our experimental results show that \sys outperforms previous CPU-based solutions by at least $11\times$ for training. 
Overall, \sys shows that GPU is also suitable to accelerate privacy-preserving DT training and inference.  
For future work, we will investigate the following directions. 

% \textcolor{black}{Add a sentence/paragraph: for the optimization in the perspective of CUDA programming, \eg cooperative groups among different blocks and parallelism optimizations, we leave this to future work.}

    \noindent{\textbf{ORAM-based Array Access.}} ORAM exhibits sub-linear communication costs in comparison to linear scans.  Existing DORAM approaches such as Floram~\cite{doerner2017scaling} rely heavily on garbled circuits, which are unlikely to yield optimal performance gains. We will design a new GPU-friendly ORAM structure. Furthermore, FSS-based protocols are communication-efficient~\cite{jawalkar2023orca}, and we plan to use them to optimize \sys.
    % Recently, Distributed ORAM (DORAM) utilizing \textit{Function Secret Sharing (FSS)}~\cite{boyle2015function} proves to be an effective solution.
    % DORAM exhibits sub-linear communication costs in comparison to linear scans, offering a method to reduce communication overhead in \sys.

    \noindent{\textbf{A General-Purpose Framework for Various Data Types.}} \textcolor{black}{\sys primarily targets categorical features. However, when dealing with continuous features, we can either apply discretization techniques~\cite{adams2022privacy} to convert them into categorical data or utilize MPC-based permutation and sorting methods~\cite{abspoel2020secure} for direct processing. It would be valuable to investigate the GPU compatibility of these protocols within a hybrid secure DT framework and explore the potential performance gains through GPU acceleration.}

\bibliographystyle{elsarticle-num}
\bibliography{ref}

% \appendix
% \input{appendix}

% \section*{Acknowledgment}
%\bibliographystyle{IEEEtranS}


% \end{thebibliography}

\end{document}